\documentclass[journal,10pt,twocolumn]{IEEEtran}
\ifCLASSINFOpdf
\else

\fi

\usepackage{amsmath}
\usepackage{amssymb}
\usepackage{mathrsfs}
\usepackage{bbm}
\usepackage{bm}
\usepackage{amsfonts}
\usepackage{cases}
\usepackage{graphicx}
\usepackage{makecell}
\usepackage{float}
\usepackage{multirow}
\usepackage{color}
\usepackage{algorithm}
\usepackage{algpseudocode}
\usepackage{graphicx}
\usepackage{subfigure}
\usepackage{threeparttable}
\usepackage{epstopdf}
\graphicspath{{./}{Fig/}}

\newtheorem{definition}{Definition}
\newtheorem{lemma}{Lemma}

\newtheorem{corollary}{Corollary}
\newtheorem{proposition}{Proposition}
\newtheorem{theorem}{Theorem}
\newtheorem{remark}{Remark}

\newtheorem{assumption}{Assumption}

\newtheorem{problem}{Problem}
\begin{document}

\title{Computing Probabilistic Controlled Invariant Sets}
\author{Yulong Gao, Karl H. Johansson, \emph{Fellow}, IEEE and Lihua Xie, \emph{Fellow}, IEEE
\thanks{This work of Y. Gao and K. H. Johansson is supported by the Knut and Alice Wallenberg Foundation, the Swedish Strategic Research Foundation, and the Swedish Research Council.}
\thanks{Y. Gao and K. H. Johansson are with Division of Decision and Control Systems, KTH Royal Institute of Technology, Stockholm 10044, Sweden
        {\tt\small yulongg@kth.se, kallej@kth.se}}%
\thanks{Y. Gao and L. Xie are with
School of Electrical and Electronic Engineering, Nanyang Technological University, 639798, Singapore
        {\tt\small ygao009@e.ntu.edu.sg, elhxie@ntu.edu.sg}}
}


\maketitle

\begin{abstract}
This paper investigates stochastic invariance for control systems through probabilistic controlled invariant sets (PCISs).  As a natural complement to robust controlled invariant sets~(RCISs), we propose finite- and infinite-horizon PCISs, and explore their relation to RICSs. We design iterative algorithms to compute the PCIS within a given set. For systems with discrete spaces, the computations of the finite- and infinite-horizon PCISs at each iteration are based on linear programming and mixed integer linear programming, respectively. The algorithms are computationally tractable and terminate in a finite number of steps. For systems with continuous spaces, we show how to discretize the spaces and prove the convergence of the approximation when computing the finite-horizon PCISs. In addition, it is shown that an infinite-horizon PCIS
can be computed by the stochastic backward reachable set from the RCIS contained in it. These PCIS algorithms are applicable to practical control systems.
Simulations are given to illustrate the effectiveness of the theoretical results for motion planning.
\end{abstract}


\ifCLASSOPTIONcaptionsoff
  \newpage
\fi

\section{Introduction}
\subsection{Motivation and Related Work}
Invariance is a fundamental concept in systems and
control~\cite{Bertsekas(1972),Blanchini(1999),Blanchini(2007)}.  A controlled invariant set captures the region where the states can be maintained by some admissible control inputs. Robust controlled invariant sets (RCISs) are defined for control systems with bounded external disturbances and address the invariance despite any realization of the disturbances. In the past decades, there have been lots of research results on RCISs and their computations~\cite{Rakovic(2005),Rungger(2017),Gilbert(1991)}. This paper studies probabilistic controlled invariant sets (PCISs), which is a natural complement to RCISs suitable in many applications. A PCIS is a set within which the controller is able to keep the system state with a certain probability. Such sets not only alleviate
the inherent conservatism of RCISs by allowing probabilistic  violations but also enlarge the applications of RCISs by being able to address unbounded disturbances. The study of PCISs is motivated by safety-critical control \cite{Mitchell(2013)}, stochastic model predictive control (MPC)~\cite{Mesbah(2016),Cannon(2011)},  reliable control~\cite{Hernandez(2016),Hernandez(2017)}, and relevant applications, e.g., air traffic management systems~\cite{HuJ(2005),Ding(2013)} and motion planning~\cite{Burlet(2004)}.

A question at the heart of this paper is
 \begin{center}
  \emph{Given a set $\mathbb{Q}$ and a parameter $0\leq \epsilon\leq 1$, how to compute a set $\tilde{\mathbb{Q}}\subseteq \mathbb{Q}$ that is invariant with probability $\epsilon$?}
 \end{center}
 To the best of our knowledge, this question  has not been explored up to now.
 One essential component in iterative approaches on computing RCISs is to compute the robust backward reachable set, in which each state can be steered to the current set by an admissible input for all possible uncertainties~\cite{Rakovic(2005),Rungger(2017),Gilbert(1991)}. The PCIS computation in this paper follows the same idea, but the robust backward  reachable set is replaced with the stochastic backward reachable sets which require different mathematical tools. Some challenges related to such an approach should be highlighted: (i) how to make it tractable to compute the stochastic backward reachable set, in particular for systems with continuous spaces; (ii) how to mitigate the conservatism when characterizing the  stochastic backward reachable set subject to the prescribed probability; (iii) how to guarantee convergence of the iterations.

\begin{table*}[htbp]
 \centering
 \label{roadmap}
   \caption{Comparisons between this paper and other work}
 \begin{tabular}{c|c|c|c|c|c}
 \hline
~ & System & Invariant Set & Control & Horizon
&
Computation\\

\hline

 This paper & \makecell{Markov controlled process} & PCIS  & Yes & \makecell{Finite and \\ infinite horizons} & \makecell{Iteration based on \\ stochastic backward reachable set}\\

\hline

\cite{Pola(2006)} & \makecell{Nonlinear stochastic system} & PCIS  & Yes & \makecell{Finite and \\ infinite horizons} & No\\

\hline

\cite{Cannon(2009)} & \makecell{Linear stochastic system} & PCIS  & Yes & \makecell{One step} & \makecell{Ellipsoidal approximation}\\

\hline
\cite{Kofman(2012)} & \makecell{Linear stochastic system} & \makecell{Probabilistic  invariant set}  & No & \makecell{Infinite horizon} & \makecell{Polyhedral approximation \\ based on   Chebyshev’s inequality} \\

 \hline
\end{tabular}
\end{table*}

Controlled invariant sets have recently been extended to stochastic systems.  In  \cite{Battilotti(2003)}, a target set, which is similar to the PCIS of this paper, is used to define stabilization in probability. In \cite{Hernandez(2016)}, a reliable control set, another similar notion to a PCIS, is used to guarantee the reliability of Markov-jump
linear systems. The reliability is further studied for such systems with bounded disturbances in \cite{Hernandez(2017)}.
A definition of PCIS for nonlinear systems is provided in \cite{Pola(2006)}  by using reachability analysis. It is later applied to portfolio optimization~\cite{Pola(2012)}.
Another definition of probabilistic invariance  originates from stochastic MPC~\cite{Cannon(2009)} and captures one-step invariance.
In~\cite{Cannon(2009)}, an ellipsoidal approximation is given for linear systems with specific uncertainty structure. Similar invariant sets are used
in~\cite{Nguyen(2018)} to construct a convex lifting function for linear stochastic control systems. A definition of a probabilistic invariant set is proposed in~\cite{Kofman(2012),Kofman(2016)} for linear stochastic systems without control inputs. This definition captures the probabilistic inclusion of the state at each time instant. A recent work \cite{Hewing(2018)} explores the correspondence between probabilistic and robust invariant sets for linear systems.
 In~\cite{Kofman(2012),Kofman(2016)}, polyhedral probabilistic invariant sets are approximated by using Chebyshev's inequality for linear systems with Gaussian noise. Recursive satisfaction is usually computationally intractable for general stochastic control systems.

The results of this paper build on the above work but make significant additions and improvements. Table~I summarizes the comparison between our work and the most relevant literature. (i) All the above references focus on some specific stochastic systems (e.g., linear or one-dimensional affine nonlinear systems) or on some specific class of  stochastic disturbances (e.g., Gaussian or state-independent noise). In our model, we consider general Markov controlled processes, which  include general system dynamics and stochastic disturbances. (ii) Different from~\cite{Kofman(2012),Kofman(2016)}, our invariant sets are defined based on trajectory inclusion as in~\cite{Pola(2006)} and, particularly, incorporate control inputs constrained by a compact set. An accompanying question is how to find an admissible control input when verifying or computing a PCIS.
(iii)  The PCISs in this paper are different from the maximal probabilistic safe sets in \cite{Abate(2007)}. Every trajectory in a PCIS is required by our definition to admit the same probability level, which does not hold for the maximal probabilistic safe set. (vi) The stochastic reachability analysis studied in \cite{Abate(2007)} provides  an important tool for maximizing the probability of staying in a set. Based on this,
we compute a PCIS within a set
with a prescribed probability level. This extends the results of~\cite{Pola(2006),Abate(2007),Amin(2014)}.


\subsection{Main Contributions and Organization}
 The objective of this paper is to  provide a novel tool to analyze invariance in stochastic control systems. The contributions are summarized as follows.

 As the first contribution, we  propose two novel definitions of PCIS: $N$-step $\epsilon$-PCIS and infinite-horizon $\epsilon$-PCIS (Definitions~\ref{finidefPCIS} and \ref{infinidefPCIS}). An $N$-step $\epsilon$-PCIS is a set within which the state can stay for $N$ steps with probability $\epsilon$ under some admissible controller while an infinite-horizon $\epsilon$-PCIS is a set within which the state can stay forever with probability $\epsilon$ under some admissible controller. These invariant sets are different from the ones proposed in \cite{Cannon(2009),Kofman(2012)}, which address probabilistic set invariance at each time
    step. Our definitions are applicable for general discrete-time stochastic control systems. We provide fundamental properties of PCISs and explore their relation to RCISs. Furthermore, we propose conditions for the existence of infinite-horizon $\epsilon$-PCIS (Theorem~\ref{infinitenecessary}).

The second contribution is that we design iterative algorithms to compute the largest finite- and infinite-horizon PCIS within a given set for systems with discrete and continuous spaces. The PCIS computation is based on the stochastic backward  reachable set. For discrete state and control spaces, it is shown that at each iteration, the  stochastic backward reachable set computation of an $N$-step $\epsilon$-PCIS can be reformulated as a linear program (LP) (Theorem~\ref{finitemaxPCIS1} and Corollary~\ref{finiteThe}) and an infinite-horizon $\epsilon$-PCIS as a computationally tractable mixed-integer linear program~(MILP) (Theorem~\ref{infinitemaximalPCIS}).
       Furthermore, we prove that these algorithms terminate in a finite number of steps. For continuous state and control spaces, we present a discretization procedure. Under weaker assumptions than~\cite{Chow(1991)}, we prove the convergence of such approximations for $N$-step $\epsilon$-PCISs (Theorem~\ref{finitediffed}). The approximations generalize the case in~\cite{Abate(2007)}, which only discretizes the state space for a given discrete control space. Furthermore, in order to compute an infinite-horizon $\epsilon$-PCIS, we propose an algorithm based on that an infinite-horizon PCIS always contains an RCIS.

%
%
%

The remainder of the paper is organized as follows. Section~\ref{preliminaries} provides the system model and some preliminaries.  Section \ref{finitePCISs} presents the definition, properties, and computation algorithms of finite-horizon PCISs. Section \ref{infinitePCISs} extends the results to the infinite-horizon case.    Examples in Section~\ref{example} illustrate the effectiveness of our approach. Section~\ref{conclusion} concludes this paper.

\textbf{Notation.} Let $\mathbb{N}$ denote the set of nonnegative integers and $\mathbb{R}$ the set of real numbers. For some $q,s \in \mathbb{N}$ and $q<s$, let $\mathbb{N}_{\geq q}$ and $\mathbb{N}_{[q,s]}$ denote the sets $\{r \in \mathbb{N}\mid r\geq q\}$ and $\{r \in \mathbb{N}\mid q \leq r\leq s\}$, respectively. For two sets $\mathbb{X}$ and $\mathbb{Y}$, $\mathbb{X}\setminus \mathbb{Y}=\{x\mid x\in\mathbb{X}, x\notin \mathbb{Y}\}$ and $\mathbb{X}\bigtriangleup \mathbb{Y}=(\mathbb{X}\setminus \mathbb{Y}) \cup (\mathbb{Y}\setminus \mathbb{X})$. When $\leq$, $\geq$, $<$, and $>$ are applied to vectors, they are interpreted element-wise. $\rm{Pr}$  denotes the probability. For a set $\mathbb{X}$, $\mathcal{B}(\mathbb{X})$ and $\mathcal{P}(\mathbb{X})$ denote the Boreal $\sigma$-algebra generated by $\mathbb{X}$ and the space of probability distributions on $\mathbb{X}$, respectively.
The indicator function of a set $\mathbb{X}$ is denoted by $\mathbbm{1}_{\mathbb{X}}(x)$, that is, if $x\in  \mathbb{X}$,
$\mathbbm{1}_{\mathbb{X}}(x)=1$  and otherwise, $\mathbbm{1}_{\mathbb{X}}(x)=0$.

\section{System Description and Preliminaries}\label{preliminaries}
Consider a stochastic control system described by a Markov controlled process  $\mathcal{S}=(\mathbb{X},\mathbb{U},T)$, where
\begin{itemize}
  \item $\mathbb{X}$ is a state space endowed with a Borel $\sigma$-algebra $\mathcal{B}(\mathbb{X})$;
  \item $\mathbb{U}$ is a compact control space endowed with a Borel $\sigma$-algebra $\mathcal{B}(\mathbb{U})$;
  \item $T:  \mathcal{B}(\mathbb{X})\times\mathbb{X}\times \mathbb{U}\rightarrow \mathbb{R}$ is a Borel-measurable stochastic kernel given $\mathbb{X}\times\mathbb{U}$, which assigns to each $x\in \mathbb{X}$ and $u\in \mathbb{U}$ a probability measure on the Borel space $(\mathbb{X},\mathcal{B}(\mathbb{X}))$: $T(\cdot| x,u)$.
\end{itemize}


Let us denote by $\mathbb{U}_x$ the set of the admissible control actions for each $x\in \mathbb{X}$. Assume that $\mathbb{U}_x$ is nonempty for each $x\in\mathbb{X}$.

Consider a finite horizon $N\in \mathbb{N}$. A policy is said to be a Markov policy if the control inputs are only dependent on the current state, i.e., $u_k=\mu_k(x_k)$.
\begin{definition}(Markov Policy)
A Markov policy $\bm{\mu}$ for system $\mathcal{S}$ is a sequence $\bm{\mu}=(\mu_0,\mu_1,\ldots, \mu_{N-1})$ of universally measurable maps
  \begin{eqnarray*}
  \mu_k:\mathbb{X}\rightarrow \mathbb{U}, \forall k\in \mathbb{N}_{[0,N-1]}.
  \end{eqnarray*}
\end{definition}

\begin{remark}
Given a space $\mathbb{Y}$, a subset $\mathbb{A}$ in this space  is universally measurable if it is measurable with respect to every complete probability measure on $\mathbb{Y}$ that measures all Borel sets in $\mathcal{B}(\mathbb{Y})$.
A function $\mu:\mathbb{Y}\rightarrow \mathbb{W}$ is universally measurable if $\mu^{-1}(\mathbb{A})$ is universally measurable in $\mathbb{Y}$ for every~$\mathbb{A}\in \mathcal{B}(\mathbb{W})$.
  As stated in~\cite{Abate(2007),Bertsekas(2004)}, the condition of universal measurability is weaker than the condition of Borel measurability for showing the existence of a solution to a stochastic optimal problem. Roughly speaking, this is because the projections of measurable sets are analytic sets and analytic sets are universally measurable but not always Borel measurable \cite{Bertsekas(2004),Stinchcombe(1992)}.
\end{remark}

\begin{remark}
For a large class of stochastic optimal control problems, Markov policies are sufficient to characterize the optimal policy~\cite{Bertsekas(2004)}.  Furthermore, since a randomized Markov policy does not increase the largest probability that the states remain in a set, we focus on deterministic Markov policies in the following.
\end{remark}

We denote the set of Markov policies as $\mathcal{M}$. Consider a set~$\mathbb{Q}\in \mathcal{B}(\mathbb{X})$. Given an initial state $x_0\in \mathbb{X}$ and a Markov policy~$\bm{\mu}\in \mathcal{M}$, an execution is a sequence of states $(x_0,x_1,\ldots, x_N)$.  Introduce the probability with which the state $x_k$ will remain within $\mathbb{Q}$ for all $k\in \mathbb{N}_{[0,N]}$:
\begin{eqnarray*}
p_{N,\mathbb{Q}}^{\bm{\mu}}(x_0)={\rm{Pr}}\{\forall k\in \mathbb{N}_{[0,N]}, x_k\in\mathbb{Q} \}.
\end{eqnarray*}

Let $p^{*}_{N,\mathbb{Q}}(x)=\sup_{\bm{\mu}\in \mathcal{M}}p_{N,\mathbb{Q}}^{\bm{\mu}}(x)$, $\forall x\in \mathbb{Q}$. We call  $p^{*}_{N,\mathbb{Q}}(x)$  the $N$-step invariance probability at $x$ in the set~$\mathbb{Q}$.  Following the dynamic program (DP) in~\cite{Abate(2007)}, define the value function $V^*_{k,\mathbb{Q}}: \mathbb{X}\rightarrow [0,1], k=0,1,\ldots,N$, by the backward recursion:
  \begin{eqnarray}\label{Vk}
  V_{k,\mathbb{Q}}^*(x)=\sup_{u\in \mathbb{U}}\mathbbm{1}_{\mathbb{Q}}(x)\int_{\mathbb{Q}}V^*_{k+1,\mathbb{Q}}(y)T(dy|x,u), x\in\mathbb{X},
  \end{eqnarray}
   with initialization
 $V_{N,\mathbb{Q}}^*(x)=1, x\in\mathbb{Q}$.

 \begin{assumption}\label{finiteassum}
   The set $$\mathbb{U}_k(x,\lambda)=\left\{u\in \mathbb{U}\mid \int_{\mathbb{X}}V^*_{k+1,\mathbb{Q}}(y)T(dy|x,u)\geq \lambda\right\}$$ is compact for all $x\in \mathbb{Q}$, $\lambda\in \mathbb{R}$, and $k\in \mathbb{N}_{[0,N-1]}$.
 \end{assumption}

\begin{lemma}\label{TheAbata1}\cite{Abate(2007)}
For all $x\in\mathbb{Q}$,
$p^{*}_{N,\mathbb{Q}}(x)=V_{0,\mathbb{Q}}^*(x)$.
   If Assumption~\ref{finiteassum} holds,
  the optimal Markov policy $\bm{\mu}_{\mathbb{Q}}^*=(\mu_{0,\mathbb{Q}}^*,\mu^*_{1,\mathbb{Q}},\ldots,\mu^*_{N-1,\mathbb{Q}})$ exists and is given by
  \begin{eqnarray*}
  \mu_{k,\mathbb{Q}}^*(x)=\arg\sup_{u\in\mathbb{U}}\mathbbm{1}_{\mathbb{Q}}(x)\int_{\mathbb{Q}}V^*_{k+1,\mathbb{Q}}(y)T(dy|x,u),
  \nonumber \\
   x\in\mathbb{Q}, k\in \mathbb{N}_{[0,N-1]}.
  \end{eqnarray*}
\end{lemma}

%

Extending the finite horizon to infinite horizon, we need to introduce stationary Markov policies.
\begin{definition}
  (Stationary Markov Policy) A Markov policy $\bm{\mu}\in \mathcal{M}$ is said to be stationary if $\bm{\mu}=(\bar{\mu},\bar{\mu},\ldots)$ with $\bar{\mu}: \mathbb{X}\rightarrow \mathbb{U}$ universally measurable.
\end{definition}

 Given an initial state $x_0\in \mathbb{X}$ and a stationary Markov policy $\bm{\mu}\in \mathcal{M}$, an execution is denoted by a sequence of states $(x_0,x_1,\ldots)$.  We introduce the probability with which the state $x_k$ will remain within $\mathbb{Q}$ for all $k\in \mathbb{N}_{\geq 0}$:
\begin{eqnarray*}
p_{\infty,\mathbb{Q}}^{\bm{\mu}}(x_0)={\rm{Pr}}\{\forall k\in \mathbb{N}, x_k\in\mathbb{Q}\}.
\end{eqnarray*}

Denote $p^{*}_{\infty,\mathbb{Q}}(x_0)=\sup_{\bm{\mu}\in \mathcal{M}}p_{\infty,\mathbb{Q}}^{\bm{\mu}}(x_0)$. We call  $p^{*}_{\infty,\mathbb{Q}}(x)$ the infinite-horizon invariance probability at $x$ in the set~$\mathbb{Q}$. Define the
  value function $G^*_{k,\mathbb{Q}}: \mathbb{X}\rightarrow [0,1], k\in \mathbb{N}_{\geq 0}$, through the forward recursion:
  \begin{eqnarray}\label{Ginf1}
   G^*_{k+1,\mathbb{Q}}(x)=\sup_{u\in \mathbb{U}}\mathbbm{1}_{\mathbb{Q}}(x)\int_{\mathbb{Q}}G^*_{k,\mathbb{Q}}(y)T(dy|x,u), x\in\mathbb{X},
  \end{eqnarray}
  initialized with
 $G^*_{0,\mathbb{Q}}(x)=1, x\in\mathbb{Q}$.

  \begin{assumption}\label{infiniteassum}
    There exists a $\bar{k}\geq 0$ such that the set $$\mathbb{U}_k(x,\lambda)=\left\{u\in \mathbb{U}\mid \int_{\mathbb{X}}G^*_{k,\mathbb{Q}}(y)T(dy|x,u)\geq \lambda\right\}$$ is compact for all $x\in \mathbb{Q}$, $\lambda\in \mathbb{R}$, and $k\in \mathbb{N}_{\geq \bar{k}}$.
  \end{assumption}
\begin{lemma}\label{TheAbata2}\cite{Abate(2007)}
   Suppose that Assumption~\ref{infiniteassum} holds.
    Then, for all  $x\in \mathbb{Q}$, the limit $G^*_{\infty,\mathbb{Q}}(x)$ exists and satisfies
  \begin{eqnarray}\label{Ginf}
  G^*_{\infty,\mathbb{Q}}(x)=\sup_{u\in \mathbb{U}}\mathbbm{1}_{\mathbb{Q}}(x)\int_{\mathbb{Q}}G^*_{\infty,\mathbb{Q}}(y)T(dy|x,u)),
  \end{eqnarray}
and  $p^{*}_{\infty,\mathbb{Q}}(x)=G^*_{\infty,\mathbb{Q}}(x)$.
   Furthermore, an optimal stationary Markov policy $\bm{\mu}_{\mathbb{Q}}^*=(\bar{\mu}_{\mathbb{Q}}^*,\bar{\mu}_{\mathbb{Q}}^*,\ldots)$ exists and is given by
    \begin{eqnarray*}
  \bar{\mu}_{\mathbb{Q}}^*(x)=\arg\sup_{u\in\mathbb{U}}\mathbbm{1}_{\mathbb{Q}}(x)
  \int_{\mathbb{Q}}G^*_{\infty,\mathbb{Q}}(y)T(dy|x,u), x\in\mathbb{Q}.
  \end{eqnarray*}
\end{lemma}

In the following two sections, we explore finite- and infinite-horizon PCISs and how to compute them.
%
%

\section{Finite-Horizon $\epsilon$-PCIS}\label{finitePCISs}
In this section, we first define finite-horizon $\epsilon$-PCIS for the system $\mathcal{S}$ and provide the properties of this set. Then, we explore how to  compute the finite-horizon $\epsilon$-PCIS within a given set.


\begin{definition}\label{finidefPCIS}
  ($N$-step $\epsilon$-PCIS) Consider a stochastic control system $\mathcal{S}=(\mathbb{X},\mathbb{U},T)$. Given a confidence level $0\leq\epsilon\leq1$, a set $\mathbb{Q}\in \mathcal{B}(\mathbb{X})$ is an $N$-step $\epsilon$-PCIS for $\mathcal{S}$ if  for any $x\in \mathbb{Q}$, there exists at least one Markov policy~$\bm{\mu}\in\mathcal{M}$ such that~$p_{N,\mathbb{Q}}^{\bm{\mu}}(x)\geq \epsilon$. 
\end{definition}

We define the stochastic backward reachable set $\mathbb{S}^{*}_{\epsilon,N}(\mathbb{Q})$ by collecting all the states $x\in \mathbb{Q}$ at which the $N$-step invariance  probability $p^*_{N,\mathbb{Q}}(x)\geq \epsilon$, i.e.,
\begin{eqnarray*}\label{finiteSBRS}
 &&\hspace{0cm}\mathbb{S}^{*}_{\epsilon,N}(\mathbb{Q})=\{x\in\mathbb{Q}\mid \exists \bm{\mu}\in \mathcal{M}, p_{N,\mathbb{Q}}^{\bm{\mu}}(x)\geq \epsilon\}\nonumber\\
&&\hspace{1.35cm} =\{x\in\mathbb{Q}\mid \sup_{\bm{\mu}\in \mathcal{M}}p_{N,\mathbb{Q}}^{\bm{\mu}}(x)\geq \epsilon\}\nonumber\\
&&\hspace{1.33cm} =\{x\in\mathbb{Q}\mid V^*_{0,\mathbb{Q}}(x)\geq \epsilon\}.
 \end{eqnarray*}

If $\mathbb{S}^{*}_{\epsilon,N}(\mathbb{Q})=\mathbb{Q}$, it yields from $\mathbb{Q}\in \mathcal{B}(\mathbb{X})$ that $\mathbb{S}^{*}_{\epsilon,N}(\mathbb{Q})$ is also Borel-measurable.
If  $\mathbb{S}^{*}_{\epsilon,N}(\mathbb{Q})\subset\mathbb{Q}$, the following lemma addresses the measurability of the set $\mathbb{S}^{*}_{\epsilon,N}(\mathbb{Q})$.
\begin{lemma}\label{finiteuniver}
For any  $\mathbb{Q}\in \mathcal{B}(\mathbb{X})$, the set $\mathbb{S}^{*}_{\epsilon,N}(\mathbb{Q})\subseteq\mathbb{Q}$ is universally measurable.
\end{lemma}
\begin{proof}
  See Appendix A.
\end{proof}

Let us denote by $\mathcal{P}(\mathbb{X})$ the set of all probability measures on $\mathbb{X}$. The following proposition shows that despite of the universal measurability of $\mathbb{S}^{*}_{\epsilon,N}(\mathbb{Q})$,   for any probability measure on $\mathbb{X}$, one can find another Borel-measurable set $\tilde{\mathbb{S}}^{*}_{\epsilon,N}(\mathbb{Q}))$ for which the difference to $\mathbb{S}^{*}_{\epsilon,N}(\mathbb{Q})$ is measure-zero.
\begin{proposition}\label{borelSBRS}
For any  $\mathbb{Q}\in \mathcal{B}(\mathbb{X})$ and   any $p\in \mathcal{P}(\mathbb{X})$, there exists a set $\tilde{\mathbb{S}}^{*}_{\epsilon,N}(\mathbb{Q})\in \mathcal{B}(\mathbb{X})$ with $\tilde{\mathbb{S}}^{*}_{\epsilon,N}(\mathbb{Q})\subseteq\mathbb{Q}$ such that $p(\tilde{\mathbb{S}}^{*}_{\epsilon,N}(\mathbb{Q})\bigtriangleup \mathbb{S}^{*}_{\epsilon,N}(\mathbb{Q}))=0$.
\end{proposition}
\begin{proof}
  It follows from the universal measurability of $\mathbb{S}^{*}_{\epsilon,N}(\mathbb{Q})$  as shown in Lemma~\ref{finiteuniver}, the Borel measurability of $\mathbb{Q}$, $\mathbb{S}^{*}_{\epsilon,N}(\mathbb{Q})\subseteq \mathbb{Q}$, and Lemma 7.26 in \cite{Bertsekas(2004)}.
\end{proof}

From Lemma~\ref{TheAbata1} and the definition of $\mathbb{S}^{*}_{\epsilon,N}(\mathbb{Q})$, we can verify whether a set $\mathbb{Q}\in \mathcal{B}(\mathbb{X})$  is an $N$-step $\epsilon$-PCIS or not by checking if either $\mathbb{S}^{*}_{\epsilon,N}(\mathbb{Q})=\mathbb{Q}$, or  $V_{0,\mathbb{Q}}^*(x)\geq \epsilon$,  $\forall x\in\mathbb{Q}$, where $V_{0,\mathbb{Q}}^*(x)$ is defined in (\ref{Vk}).

\begin{remark}\label{maximalsafeset}
  The stochastic backward reachable set $\mathbb{S}^{*}_{\epsilon,N}(\mathbb{Q})$ is called the maximal probabilistic safe set in \cite{Abate(2007)}. The $N$-step $\epsilon$-PCIS $\mathbb{Q}$ in Definition \ref{finidefPCIS} refines the maximal probabilistic safe set by requiring that for any initial state $x_0\in\mathbb{Q}$, the $N$-step invariance probability $p_{\infty,\mathbb{Q}}^{*}(x_0)$ is no less than $\epsilon$.
\end{remark}

 In the following, we show that finite-horizon PCISs  are closed under union.

\begin{proposition}\label{finiteproperty}
   Consider a collection of sets $\mathbb{Q}_i\in \mathcal{B}(\mathbb{X})$, $i=1, \ldots, r$. If each $\mathbb{Q}_i$ is an $N_i$-step $\epsilon_i$-PCIS for the same system $\mathbb{S}$, then the union  $\bigcup_{i=1}^r \mathbb{Q}_i$ is an $N$-step $\epsilon$-PCIS, where $N=\min_{i}N_i$ and $\epsilon=\min_{i}\epsilon_i$.
\end{proposition}
\begin{proof}
The result follows from the following two facts:\\
(i) for any $\mathbb{Q},\mathbb{P}\in \mathcal{B}(\mathbb{X})$ with $\mathbb{Q}\subseteq \mathbb{P}$,   $\sup_{\bm{\mu}\in \mathcal{M}}p_{N,\mathbb{Q}}^{\bm{\mu}}(x)\leq \sup_{\bm{\mu}\in \mathcal{M}}p_{N,\mathbb{P}}^{\bm{\mu}}(x)$,  $\forall N\in \mathbb{N}$ and $\forall x\in \mathbb{Q}$; \\
(ii) for any  $N, N'\in \mathbb{N}$  with  $N\leq N'$, $\sup_{\bm{\mu}\in \mathcal{M}}p_{N',\mathbb{Q}}^{\bm{\mu}}(x)\leq \sup_{\bm{\mu}\in \mathcal{M}}p_{N,\mathbb{Q}}^{\bm{\mu}}(x)$,  $\forall Q\in \mathcal{B}(\mathbb{X})$ and $\forall x\in \mathbb{Q}$.
\end{proof}

\vspace{-0.1cm}
\subsection{Finite-horizon $\epsilon$-PCIS computation}
This subsection will address the following problem.
\begin{problem}\label{finiteproblem}
Given a set $\mathbb{Q}\in \mathcal{B}(\mathbb{X})$ and a prescribed probability $0\leq \epsilon \leq 1$, compute an $N$-step $\epsilon$-PCIS $\tilde{\mathbb{Q}}\subseteq\mathbb{Q}$.
\end{problem}

To handle this problem, our basic idea is to iteratively compute stochastic backward reachable sets until convergence.
A general procedure is presented in the following algorithm.
 \begin{algorithm}[H]\label{algorithm1}
\caption{$N$-step $\epsilon$-PCIS}
\begin{algorithmic}[1]
\State Initialize $i=0$ and $\mathbb{P}_i=\mathbb{Q}$.
\State Compute $V^{*}_{0,\mathbb{P}_i}(x), \forall x\in \mathbb{P}_i$.
\State Compute $\mathbb{S}^{*}_{\epsilon,N}(\mathbb{P}_i)$ and construct a Borel-measurable set $\tilde{\mathbb{S}}^{*}_{\epsilon,N}(\mathbb{P}_i)$ such that $p(\tilde{\mathbb{S}}^{*}_{\epsilon,N}(\mathbb{P}_i)\bigtriangleup \mathbb{S}^{*}_{\epsilon,N}(\mathbb{P}_i))=0$ for some $p\in \mathcal{P}(\mathbb{X})$;
\State If $\mathbb{P}_{i+1}=\mathbb{P}_i$, stop. Else, set $i=i+1$ and go to step 2.
\end{algorithmic}
\end{algorithm}

In Algorithm~1, we first compute the stochastic  backward reachable set $\mathbb{S}^{*}_{\epsilon,N}(\mathbb{P}_i)$ within $\mathbb{P}_i$ and then update $\mathbb{P}_{i+1}$ to be the corresponding Borel-measurable set $\tilde{\mathbb{S}}^{*}_{\epsilon,N}(\mathbb{P}_i)$, which is tailored by picking up a $p\in \mathcal{P}(\mathbb{S})$ such that  $p(\tilde{\mathbb{S}}^{*}_{\epsilon,N}(\mathbb{P}_i)\bigtriangleup \mathbb{S}^{*}_{\epsilon,N}(\mathbb{P}_i))=0$ (see Proposition~\ref{borelSBRS}). The following theorem shows convergence of $\mathbb{P}_i$. The terminal condition guarantees that the resulting set by this algorithm is an $N$-step $\epsilon$-PCIS $\tilde{\mathbb{Q}}\subseteq\mathbb{Q}$.


\begin{theorem}\label{finitemaxPCIS1}
Let Assumption~\ref{finiteassum} hold. For any $\mathbb{Q}\in\mathcal{B}(\mathbb{X})$, Algorithm $1$ converges, i.e., $\lim_{i\rightarrow\infty}\mathbb{P}_i$ exists.  If $\lim_{i\rightarrow\infty}\mathbb{P}_i\neq \emptyset$, it is the largest $N$-step $\epsilon$-PCIS within $\mathbb{Q}$.
\end{theorem}
\begin{proof}
From Algorithm $1$ and Lemma~\ref{borelSBRS}, we have that if the termination condition does not hold, $\mathbb{P}_{i+1}\subset \mathbb{P}_i$. It follows that
the sequence $\{\mathbb{P}_i\}_{i\in \mathbb{N}}$ is nonincreasing. Then,
\begin{eqnarray*}
\liminf_{i\rightarrow\infty}\mathbb{P}_i=\bigcup\limits_{i\geq 1}\bigcap\limits_{j\geq i}\mathbb{P}_j=\bigcap\limits_{j\geq 1}\mathbb{P}_j=\bigcap\limits_{i\geq 1}\bigcup\limits_{j\geq i}\mathbb{P}_j=\limsup_{i\rightarrow\infty}\mathbb{P}_i,
\end{eqnarray*}
which suggests the existence of $\lim_{i\rightarrow\infty}\mathbb{P}_i$.
Furthermore, if $\lim_{i\rightarrow\infty}\mathbb{P}_i$ is nonempty, we conclude that it is  the largest $N$-step PCIS within $\mathbb{Q}$ based on the fixed-point theory.
\end{proof}

To facilitate the practical implementation of Algorithm 1, we need to address two important properties: the computational tractability  of $V^*_{0,\mathbb{P}_i}(x)$, $\forall x\in \mathbb{P}_i$, and the finite-step convergence of Algorithm 1. In the following, we will derive these two properties for discrete and continuous spaces, respectively. It is shown that if the spaces are discrete, the properties are guaranteed and in particular
at each iteration we only need to solve an LP to compute the exact value of $V^*_{0,\mathbb{P}_i}$. If the spaces are continuous, we will design a discretization algorithm with convergence guarantee, which enables us to preserve the above two properties.

\subsubsection{Discrete state and control spaces}\label{finitediscretespace}
If the state and control spaces are discrete, i.e., they are finite sets,  the stochastic kernel $T(y|x,u)$ denotes the transition probability from  state $x\in \mathbb{X}$ to  state $y\in \mathbb{X}$ under  control action $u\in\mathbb{U}_x$, which satisfies that $\sum_{y\in\mathbb{X}}T(y|x,u)=1$, $\forall x\in \mathbb{X}$ and $u\in\mathbb{U}_x$.

In this case, according to Theorem 1 of~\cite{Bhattacharya(2017)}, we can exactly compute $V_{0,\mathbb{P}_i}^*(x)$ via an LP. Moreover, the existence of the optimal Markov policy can be always guaranteed.

\begin{lemma}\label{lemmafinite}
 Given any set $\mathbb{P}_i\subset \mathbb{X}$, the value functions $V_{k,\mathbb{P}_i}^*$ in (\ref{Vk}) can be obtained by solving an LP:
\begin{subequations}\label{LPfinite}
  \begin{eqnarray}
  && \hspace{-1.2cm}\min \ \sum_{k=0}^N\sum_{x\in \mathbb{P}_i}v_k(x) \\
  && \hspace{-1.2cm}\text{subject to}\ \forall x\in \mathbb{P}_i \nonumber \\
  && \hspace{-1.2cm}
      v_k(x)\geq \sum_{y\in \mathbb{P}_i}v_{k+1}(y)T(y|x,u), \forall u\in \mathbb{U}_x, \forall k\in \mathbb{N}_{[0,N-1]},\\
  && \hspace{-1.2cm}
  v_N(x)\geq 1,
   \end{eqnarray}
\end{subequations}
which gives $V^*_{k,\mathbb{P}_i}(x)=v^*_k(x)$, $\forall x\in\mathbb{P}_i$ and $\forall k\in \mathbb{N}_{[0,N]}$, where $v^*_k$ is the optimal solution of  (\ref{LPfinite}). The optimal Markov policy $\bm{\mu}_{\mathbb{P}_i}^*=(\mu_{0,\mathbb{P}_i}^*,\mu^*_{1,\mathbb{P}_i},\ldots,\mu^*_{N-1,\mathbb{P}_i})$ is given by $\mu^*_{k,\mathbb{P}_i}(x)=u$ where $u\in \mathbb{U}_x$ is such that
  \begin{eqnarray}\label{dualLPfinite}
  v_k^*(x)= \sum_{y\in \mathbb{P}_i}v^*_{k+1}(y)T(y|x,u).
   \end{eqnarray}
\end{lemma}
\begin{proof}
  See Theorem 1 in \cite{Bhattacharya(2017)} for the proof.
\end{proof}

\begin{corollary}\label{finiteThe}
  For discrete state and control spaces, Algorithm~1 converges in a finite number of iterations.  Furthermore, at each iteration, the $N$-step invariance probability $V^*_{0,\mathbb{P}_i}(x)$, $\forall x\in \mathbb{P}_i$, can be computed via the LP (\ref{LPfinite}) and the corresponding optimal policy is determined by (\ref{dualLPfinite}).
\end{corollary}
\begin{proof}
The finite-step convergence  of Algorithm~1 follows from Theorem~\ref{finitemaxPCIS1} and the finite cardinality of~$\mathbb{Q}$.  The remaining part follows from  Lemma~\ref{lemmafinite}.
\end{proof}

\begin{remark}\label{Remark:Alg1Complexity}
When implementing Algorithm 1 to a system with discrete spaces, the maximal  number of iterations is $|\mathbb{Q}|$. At each iteration, an LP is solved to compute the value of $V^*_{0,\mathbb{P}_i}(x)$, $\forall x\in \mathbb{P}_i$. The number of the decision values in the LP is at most $|\mathbb{Q}|(N+1)$ and the number of  constraints is at most $|\mathbb{Q}|(N |\mathbb{U}|+1)$. It follows from  \cite{Megiddo(1984)} that Algorithm 1 can be implemented in $O(|\mathbb{Q}|^2(N |\mathbb{U}|+1))$ time.
\end{remark}

\subsubsection{Continous state and control action spaces}
In order to preserve the computational tractability of $V^*_{0,\mathbb{P}_i}$ and the finite-step convergence of Algorithm 1, if  the state and control spaces are both continuous, we first discretize the spaces with convergence guarantee. Then, we adapt Algorithm 1 to compute an approximate $N$-step $\epsilon$-PCIS  within a given set.

Assume that $\mathbb{X}\subseteq \mathbb{R}^{n_x}$ and $\mathbb{U}\subseteq \mathbb{R}^{n_u}$ for some $n_x,n_u\in\mathbb{N}$.  For simplicity, we use Euclidean metric for the spaces $\mathbb{X}$ and~$\mathbb{U}$. For any $\mathbb{Q}\in  \mathcal{B}(\mathbb{X})$, we define $\phi(\mathbb{Q})=Leb(\mathbb{Q})$ where $Leb(\cdot)$ denotes the Lebesgue measure of sets.  We suppose that the stochastic kernel $T(\cdot|x,u)$ admits a density $t(y|x,u)$, which represents the probability density of $y$ given the current state $x$ and the control action $u$.

Now we consider Problem \ref{finiteproblem}, where we assume that the given set $\mathbb{Q}\in \mathcal{B}(\mathbb{X})$ is compact, which implies that $\phi(\mathbb{Q})$ is bounded. We further suppose that the density function satisfies the following assumption.

\begin{assumption}\label{assumt}
There exists a constant $L$ such that for any $x,x',y,y'\in \mathbb{Q}$, and $u,u'\in\mathbb{U}$,  $$|t(y|x,u)-t(y'|x',u')|\leq L(\|y-y'\|+\|x-x'\|+\|u-u'\|).$$
\end{assumption}

\paragraph{Discretization}
We discretize the compact set $\mathbb{Q}\subset \mathbb{X}$ into $m_x$  pair-wise disjoint nonempty Borel sets $\mathbb{Q}_{i}$,  $i\in \mathbb{N}_{[1,m_x]}$, i.e.,  $\mathbb{Q}=\cup_{i=1}^{m_x}\mathbb{Q}_{i}$. We pick a representative state from each set $\mathbb{Q}_{i}$, denoted by $q_i$. Let $\hat{\mathbb{Q}}=\{q_i,i\in\mathbb{N}_{[1,m_x]}\}$,  $d_i=\sup_{x,y\in \mathbb{Q}_i}\|x-y\|$, and  $D_x=\max_{i\in\mathbb{N}_{[1,m_x]}}d_i$.

Similarly, the compact control space $\mathbb{U}$ is divided into  $m_u$ pair-wise disjoint nonempty Borel sets  $\mathbb{C}_{i}$, $i\in\mathbb{N}_{[1,m_u]}$, i.e., $\mathbb{U}=\cup_{i=1}^{m_u}\mathbb{C}_{i}$. We pick a representative element from the set $\mathbb{C}_{i}$, denoted by $\hat{u}_i$. Let $\hat{\mathbb{U}}=\{\hat{u}_i,i\in\mathbb{N}_{[1,m_u]}\}$, $l_i=\sup_{x,y\in \mathbb{C}_i}\|x-y\|$, and $D_u=\max_{i\in\mathbb{N}_{[1,m_u]}}l_i$.

Let the grid size be a constant $\delta \geq \max\{D_x,D_u\}$.
For each $x\in \mathbb{Q}$, define the set of admissible discrete control actions as
\begin{eqnarray}\label{tildeUx}
\hat{\mathbb{U}}_x=\{\hat{u}\in\hat{\mathbb{U}}\mid \|u-\hat{u}\|\leq \delta \ \text{for some} \  u\in \mathbb{U}_{s_x}\},
\end{eqnarray}
where $s_x$ is the representative state of $\mathbb{Q}_i$  to which $x$ belongs, i.e., $s_x=q_i$ if $x\in \mathbb{Q}_i$.
Following \cite{Chow(1991)}, the following lemma  shows that each  $x\in \mathbb{Q}$ has a nonempty admissible discretized control set.

\begin{lemma}\label{discretizedU}
For each $q_i\in\hat{\mathbb{Q}}$, the set $\hat{\mathbb{U}}_{q_i}$ is nonempty and $\hat{\mathbb{U}}_x=\hat{\mathbb{U}}_{q_i}$, $\forall x\in\mathbb{Q}_i$.
\end{lemma}
\begin{proof}
Since   the admissible control set $\mathbb{U}_{s_x}$ is nonempty, $\forall x\in\mathbb{Q}$, there exists $\hat{u}\in \hat{\mathbb{U}}$ such that $\|u-\hat{u}\|\leq \delta$, $\forall u\in \mathbb{U}_{s_x}$. Hence, by the definition of $s_x$, we have that the set $\hat{\mathbb{U}}_{q_i}$ is nonempty for each $q_i\in\hat{\mathbb{Q}}$. Furthermore, from (\ref{tildeUx}), it is easy to obtain that $\hat{\mathbb{U}}_x=\hat{\mathbb{U}}_{q_i}$, $\forall x\in\mathbb{Q}_i$.
\end{proof}

As in~\cite{Chow(1991)}, let us define the function $\hat{t}: \mathbb{Q}\times \mathbb{Q}\times \hat{\mathbb{U}}\rightarrow \mathbb{R}$
\begin{eqnarray}
&&\hspace{-1cm}\hat{t}(y|x,\hat{u})= \begin{cases}
  \frac{t(s_y|s_x,\hat{u})}{\int_{\mathbb{Q}}t(s_z|s_x,\hat{u})dz}
  , & \mbox{if } \int_{\mathbb{Q}}t(s_z|s_x,\hat{u})dz\geq 1, \\
  t(s_y|s_x,\hat{u}), & \mbox{otherwise}.
\end{cases}\label{dist}
\end{eqnarray}
From (\ref{dist}), we observe that all states $y\in \mathbb{Q}_i$ enjoy the same stochastic kernel. An approximate stochastic control system is given by a triple $\hat{\mathcal{S}}_{\mathbb{Q}}=(\hat{\mathbb{Q}},\hat{\mathbb{U}},\hat{T})$. Here the transition probability $\hat{T}(q_j|q_i,\hat{u})$ is defined by
$\hat{T}(q_j|q_i,\hat{u})=\int_{\mathbb{Q}_j}\hat{t}(y|q_i,\hat{u})dy$,
where $q_i,q_j\in \hat{\mathbb{Q}}$ with $q_i\in\mathbb{Q}_i$ and $q_j\in\mathbb{Q}_j$, and $\hat{u}\in \hat{\mathbb{U}}$.
\paragraph{Approximation of PCISs}
For the approximate system $\hat{\mathcal{S}}_{\mathbb{Q}}$,
the discretized version of the DP (\ref{Vk}) is given by
\begin{eqnarray*}\label{disV}
\begin{cases}
\hat{V}^*_{N,\mathbb{Q}}(q_i)=1, \\
\hat{V}^*_{k,\mathbb{Q}}(q_i)=\max\limits_{\hat{u}\in\hat{\mathbb{U}}}\big(\sum\limits_{j=1}^{m_x}
\hat{V}^*_{k+1,\mathbb{Q}}(q_j)\hat{T}(q_j|q_i,\hat{u})\big), \forall k\in\mathbb{N}_{[0,N-1]}.
\end{cases}
\end{eqnarray*}
For each $x\in \mathbb{Q}_i$, $\hat{V}^*_{k,\mathbb{Q}}(x)=\hat{V}_{k,\mathbb{Q}}^*(q_i),  \forall k\in\mathbb{N}_{[0,N]}$.
 We define the discretized optimal Markov policy $\hat{\bm{\mu}}_{\mathbb{Q}}^*=(\hat{\mu}^*_{0,\mathbb{Q}},\ldots,\hat{\mu}^*_{N-1,\mathbb{Q}})$ as
\begin{eqnarray*}\label{disMP}
 &&\hspace{-0.3cm}\hat{\mu}^*_{k,\mathbb{Q}}(q_i)=\arg \max_{\hat{u}\in \hat{\mathbb{U}}}\int_{\mathbb{Q}}\hat{V}^*_{k+1,\mathbb{Q}}(y)\hat{t}(y|q_i,\hat{u})dy,\nonumber \\
  &&\hspace{1cm}=\arg \max\limits_{\hat{u}\in\hat{\mathbb{U}}}\big(\sum\limits_{j=1}^{m_x}
\hat{V}^*_{k+1,\mathbb{Q}}(q_j)\hat{T}(q_j|q_i,\hat{u})\big).
\end{eqnarray*}
For each $x\in \mathbb{Q}_i$,
$\hat{\mu}^*_{k,\mathbb{Q}}(x)= \hat{\mu}^*_{k,\mathbb{Q}}(q_i), \ \forall k\in\mathbb{N}_{[0,N-1]}$.

\begin{remark}
  Since  the state and control action spaces of the approximated system $\hat{\mathcal{S}}$ are finite,   the value of $\hat{V}^*_{k,\mathbb{Q}}$ can be computed via the LP (\ref{LPfinite}) and the corresponding optimal policy can be determined by (\ref{dualLPfinite}). In addition, all the states in each $\mathbb{Q}_i$ share the same approximate $N$-step invariance probability and optimal policy as the representative state $q_i\in \mathbb{Q}_i$.
\end{remark}

\begin{lemma}\label{Thedisfinite}
  Under Assumptions~\ref{finiteassum} and \ref{assumt}, the functions $V^*_{k,\mathbb{Q}}(x)$ and  $\hat{V}^*_{k,\mathbb{Q}}(x)$ satisfy that $\forall x\in \mathbb{Q}$,
  \begin{eqnarray}
  &&|V^*_{k,\mathbb{Q}}(x)-\hat{V}^*_{k,\mathbb{Q}}(x)|\leq \tau_k(\mathbb{Q})\delta,\label{difV1}
  \end{eqnarray}
 where
 \begin{eqnarray}\label{tauk}
 \begin{cases}
 \tau_N(\mathbb{Q})=0,\\
 \tau_k(\mathbb{Q})=4\phi(\mathbb{Q})L+\tau_{k+1}(\mathbb{Q}), \ \forall k\in \mathbb{N}_{[0,N-1]}.
 \end{cases}
 \end{eqnarray}
\end{lemma}
\begin{proof}
  See Appendix B.
\end{proof}

\begin{remark}
 Lemma~\ref{Thedisfinite} guarantees  convergence as the grid size tends to zero and generalizes the case considered in \cite{Abate(2007)},
which only discretizes the state space for a given finite control space.
To  prove Lemma~\ref{Thedisfinite}, we need to show that (i) the value functions in  (\ref{Vk})  are Lipschitz
continuous (Lemma \ref{LipschitzV}), which is similar to Theorem 8 in \cite{Abate(2007)}, and (ii) the  difference between the approximate
density function and the original density function is bounded (Lemma \ref{Lipschitzt}), which is different from that in \cite{Abate(2007)}.
\end{remark}

\begin{theorem}\label{finitediffed}
Let Assumptions~\ref{finiteassum} and \ref{assumt} hold. Consider a compact set $\mathbb{Q}\in \mathcal{B}(\mathbb{X})$ and a corresponding discretized set $\hat{\mathbb{Q}}$ of $\mathbb{Q}$. If $\hat{\mathbb{Q}}$ is an $N$-step $\hat{\epsilon}$-PCIS for the approximate system $\hat{\mathcal{S}}_{\mathbb{Q}}=(\hat{\mathbb{Q}},\hat{\mathbb{U}},\hat{T})$, and $\hat{\epsilon}\geq\tau_0(\mathbb{Q})\delta $, the set $\mathbb{Q}$ is an $N$-step $\epsilon$-PCIS for the system $\mathcal{S}$,  where $\epsilon=\hat{\epsilon}-\tau_0(\mathbb{Q})\delta$.
\end{theorem}
\begin{proof}
According to the construction of the discretized system $\hat{\mathcal{S}}_{\mathbb{Q}}$, we have that $\forall k\in \mathbb{N}_{[0,N]}$, $\forall i\in \mathbb{N}_{[1,m_x]}$ and $\forall x\in \mathbb{Q}_i$, $\hat{V}^*_{k,\mathbb{Q}}(x)=\hat{V}^*_{k,\mathbb{Q}}(q_i)$. Since $\hat{\mathbb{Q}}$ is an $N$-step $\hat{\epsilon}$-PCIS, it follows that $\forall x\in \mathbb{Q}$, $\hat{V}^*_{0,\mathbb{Q}}(x)\geq \hat{\epsilon}$. By Lemma~\ref{Thedisfinite} and triangular inequality, we have
  \begin{eqnarray*}
  V^*_{0,\mathbb{Q}}(x)\geq \hat{V}^*_{0,\mathbb{Q}}(x)- \tau_0(\mathbb{Q})\delta \geq \hat{\epsilon}-\tau_0(\mathbb{Q})\delta, \forall x\in \mathbb{Q}.
  \end{eqnarray*}
  Then, when $\hat{\epsilon}\geq\tau_0(\mathbb{Q})\delta$, we conclude that the set $\mathbb{Q}$ is an $N$-step $\epsilon$-PCIS where $0\leq\epsilon=\hat{\epsilon}-\tau_0(\mathbb{Q})\delta$.
\end{proof}

\begin{remark}
  From Theorem~\ref{finitediffed}, if $0\leq \epsilon< 1$, by choosing a suitable grid size $0<\delta\leq \frac{1-\epsilon}{\tau_0(\mathbb{Q})}$, the problem of computing  an $N$-step $\epsilon$-PCIS within $\mathbb{Q}$ for $\mathcal{S}$ can be transformed into that of computing an approximate $N$-step $\hat{\epsilon}$-PCIS with probability $\hat{\epsilon} \geq\epsilon+\tau_0(\mathbb{Q})\delta$ for $\hat{\mathcal{S}}_{\mathbb{Q}}$.
\end{remark}


\paragraph{Computation algorithm}
 Assume that a probability level $0\leq \epsilon< 1$ is given. After discretizing the set $\mathbb{Q}$ and the control space $\mathbb{U}$, we modify Algorithm 1 to compute an $N$-step $\epsilon$-PCIS $\tilde{\mathbb{Q}}\subseteq\mathbb{Q}$, as shown in the following.
\begin{algorithm}[H]\label{algorithm2}
\caption{Approximate $N$-step $\epsilon$-PCIS}
\begin{algorithmic}[1]
\State Choose grid size  $0<\delta< \frac{1-\epsilon}{\tau_0(\mathbb{Q})}$, discretize the sets $\mathbb{Q}$ and $\mathbb{U}$, construct an approximate system $\hat{\mathcal{S}}_{\mathbb{Q}}=(\hat{\mathbb{Q}},\hat{\mathbb{U}},\hat{T})$.
\State Initialize $i=0$, $\mathbb{P}_i=\mathbb{Q}$, and $\hat{\mathbb{P}}_i=\hat{\mathbb{Q}}$.
\State Compute $\hat{V}^*_{0,\mathbb{P}_i}(q_j)$, $\forall q_j\in \hat{\mathbb{P}}_i$.
\State Compute $\tau_0(\mathbb{P}_i)$ by (\ref{tauk})  and $\hat{\epsilon}=\epsilon+\tau_0(\mathbb{P}_i)\delta$.
\State Compute the set $\hat{\mathbb{P}}_{i+1}=\mathbb{S}^{*}_{\hat{\epsilon},N}(\hat{\mathbb{P}}_i)$ for $\hat{\mathcal{S}}_{\mathbb{Q}}$ and $\mathbb{P}_{i}=\cup_{q_j\in\hat{\mathbb{P}}_{i}}\mathbb{Q}_j$
\State If $\hat{\mathbb{P}}_{i+1}=\hat{\mathbb{P}}_i$, stop. Else, set $i=i+1$ and go to step 3.
\end{algorithmic}
\end{algorithm}

In Algorithm 2, we first construct an approximate system $\hat{\mathcal{S}}_{\mathbb{Q}}=(\hat{\mathbb{Q}},\hat{\mathbb{U}},\hat{T})$ with grid size  $0<\delta< \frac{1-\epsilon}{\tau_0(\mathbb{Q})}$. Then, following similar steps as in Algorithm 1, we compute the stochastic backward reachable set iteratively for the system  $\hat{\mathcal{S}}_{\mathbb{Q}}$. At each iteration, an LP is solved to obtain the $N$-step invariance  probability. One difference is that the stochastic backward reachable set is computed with respect to $\hat{\epsilon}=\epsilon+\tau_0(\mathbb{P}_i)\delta$ and the updated set for the system $\mathcal{S}$ is the union of the subsets of $\mathbb{Q}$ corresponding  to the stochastic backward reachable set. By Theorem~\ref{finitediffed}, the resulting set by Algorithm 2 is an
$N$-step $\epsilon$-PCIS.

\begin{corollary}\label{finiteThe1}
Let Assumptions~\ref{finiteassum} and \ref{assumt} hold. For continuous state and control spaces, Algorithm~2 converges in a finite number of iterations and generates an
$N$-step $\epsilon$-PCIS. Furthermore, at each iteration, the $N$-step invariance probability $\hat{V}^*_{0,\mathbb{P}_i}(q_j)$, $\forall q_j\in \hat{\mathbb{P}}_i$, can be computed via the LP (\ref{LPfinite}) and the corresponding optimal policy is determined by (\ref{dualLPfinite}).
\end{corollary}
\begin{proof}
By Theorem~\ref{finitediffed} and the Borel measurability of the subsets $\mathbb{Q}_i, \forall i\in \mathbb{N}_{[1,m_x]}$, it follows that the set generated by Algorithm~2 is an  $N$-step $\epsilon$-PCIS.
The remaining part is similar to the proof of Corollary~\ref{finiteThe}.
\end{proof}

\begin{remark}
 When implementing Algorithm 2 to a system with continuous spaces,  it follows from   \cite{Megiddo(1984)} that Algorithm 2 can be implemented in $O(m^2_x(Nm_u+1))$ time, cf. Remark~\ref{Remark:Alg1Complexity}.
\end{remark}

\section{Extension to Infinite-horizon $\epsilon$-PCIS}\label{infinitePCISs}
Now let us extend finite-horizon $\epsilon$-PCISs to infinite-horizon $\epsilon$-PCISs. In this section, we define the infinite-horizon $\epsilon$-PCIS and explore the conditions of its existence. Furthermore, we provide algorithms to compute an infinite-horizon $\epsilon$-PCIS within a given set.

\begin{definition}\label{infinidefPCIS}
  (Infinite-horizon PCIS) Consider a stochastic control system $\mathcal{S}=(\mathbb{X},\mathbb{U},T)$. Given a confidence level $0\leq\epsilon\leq1$, a set $\mathbb{Q} \in \mathcal{B}(\mathbb{X})$ is an infinite-horizon $\epsilon$-PCIS for $\mathcal{S}$ if for any $x\in \mathbb{Q}$, there exists at least one stationary Markov policy $\bm{\mu}\in\mathcal{M}$ such that $p_{\infty,\mathbb{Q}}^{\bm{\mu}}(x)\geq \epsilon$.
\end{definition}

We define the stochastic backward reachable set $\mathbb{S}^{*}_{\epsilon,\infty}(\mathbb{Q})$ by collecting all the states $x\in \mathbb{Q}$ at which the infinite-horizon invariance  probability $p^*_{\infty,\mathbb{Q}}(x)\geq \epsilon$, i.e.,
\begin{eqnarray*}\label{infiniteSBRS}
 &&\hspace{0cm}\mathbb{S}^{*}_{\epsilon,\infty}(\mathbb{Q})=\{x\in\mathbb{Q}\mid \exists \bm{\mu}\in \mathcal{M}, p_{\infty,\mathbb{Q}}^{\bm{\mu}}(x)\geq \epsilon\}\nonumber\\
&&\hspace{1.35cm} =\{x\in\mathbb{Q}\mid \sup_{\bm{\mu}\in \mathcal{M}}p_{\infty,\mathbb{Q}}^{\bm{\mu}}(x)\geq \epsilon\}\nonumber\\
&&\hspace{1.33cm} =\{x\in\mathbb{Q}\mid G^*_{\infty,\mathbb{Q}}(x)\geq \epsilon\}.
 \end{eqnarray*}

For the infinite-horizon case, Lemma~\ref{finiteuniver} and Proposition~\ref{borelSBRS} still hold. That is, the set $\mathbb{S}^{*}_{\epsilon,\infty}(\mathbb{Q})$ is universally measurable and for any $p\in \mathcal{P}(\mathbb{S})$, there exists another Borel-measurable set $\tilde{\mathbb{S}}^{*}_{\epsilon,\infty}(\mathbb{Q})\subseteq \mathbb{Q}$ such that $p(\tilde{\mathbb{S}}^{*}_{\epsilon,\infty}(\mathbb{Q})\bigtriangleup \mathbb{S}^{*}_{\epsilon,\infty}(\mathbb{Q}))=0$.

Under Assumption~\ref{infiniteassum},
by Lemma~\ref{TheAbata2} and the definition of $\mathbb{S}^{*}_{\epsilon,\infty}(\mathbb{Q})$, we can verify whether a set $\mathbb{Q}\in \mathcal{B}(\mathbb{X})$  is an infinite-horizon $\epsilon$-PCIS or not by checking if either $\mathbb{S}^{*}_{\epsilon,\infty}(\mathbb{Q})=\mathbb{Q}$, or  $G_{\infty,\mathbb{Q}}^*(x)\geq \epsilon$,  $\forall x\in\mathbb{Q}$, where  $G_{\infty,\mathbb{Q}}^*(x)$ is defined by (\ref{Ginf1})--(\ref{Ginf}).

\begin{definition}\label{defRCIS}
Consider a stochastic control system $\mathcal{S}=(\mathbb{X},\mathbb{U},T)$. An RCIS $\mathbb{Q} \in \mathcal{B}(\mathbb{X})$ for $\mathcal{S}$ is an $N$-step $\epsilon$-PCIS with $N=1$ and $\epsilon=1$.
\end{definition}

\begin{remark}
Another interpretation of RCIS in Definition~\ref{defRCIS} is that a set $\mathbb{Q} \in \mathcal{B}(\mathbb{X})$ is an RCIS if for any $x\in \mathbb{Q}$, there exists at least one control input $u\in \mathbb{U}$ such that $T(\mathbb{Q}|x,u)=1$.
 It is easy to verify that an RCIS is also an infinite-horizon $\epsilon$-PCIS with $\epsilon=1$. It is called an absorbing set in \cite{Tkachev(2011)} where there is no control input. In the following, we show that the RCIS plays an important role in the existence of infinite-horizon PCIS and provide how to design an algorithm to compute such PCIS based on RCIS.
\end{remark}

\begin{remark}
Note that infinite-horizon $\epsilon$-PCISs are also closed under union, as shown in Proposition~\ref{finiteproperty} when $N$ is replaced by $\infty$.
\end{remark}

\subsection{Existence of infinite-horizon PCIS}
Intuitively, the monotone decrease of $G^*_{\infty,\mathbb{Q}}(x)$ may imply that the value of $G^*_{\infty,\mathbb{Q}}(x)$ is one or zero. However, it is possible to get $0<G^*_{\infty,\mathbb{Q}}(x)< 1$ in some cases (see Examples 1 and 2 in Section \ref{example}).
The following theorem provides necessary conditions and sufficient conditions for the existence of infinite-horizon $\epsilon$-PCIS with $\epsilon>0$.

\begin{theorem}\label{infinitenecessary}
Suppose that Assumption~\ref{infiniteassum} holds and let $0< \epsilon\leq1$ be fixed. Given a nonempty set $\mathbb{Q}$, let $u_x$ be the control input such that (\ref{Ginf}) holds for each $x\in \mathbb{Q}$. The set $\mathbb{Q}$ is an infinite-horizon $\epsilon$-PCIS
\begin{itemize}
  \item[(i)] \emph{only if} there exists an RCIS $\mathbb{Q}_f\subseteq \mathbb{Q}$ such that $\forall x\in \mathbb{Q}\setminus\mathbb{Q}_f$,
      \begin{eqnarray}\label{Eq:Ginfupperbound}
     &&\hspace{-1cm} T(\mathbb{Q}_f|x,u_x)
    +\int_{\mathbb{Q}\setminus\mathbb{Q}_f}T(\mathbb{Q}_f|y,u_{y})T(dy|x,u_x)\nonumber\\
   &&\hspace{4cm} +\frac{\rho^2}{1-\rho}\geq \epsilon,
      \end{eqnarray}
   where $\rho=\sup_{x\in \mathbb{Q}\setminus\mathbb{Q}_f}\int_{\mathbb{Q}\setminus\mathbb{Q}_f}T(dy|x,u_x)$;
  \item[(ii)] \emph{if }there exists an RCIS $\mathbb{Q}_f\subseteq \mathbb{Q}$ such that $\forall x\in \mathbb{Q}\setminus\mathbb{Q}_f$,
      \begin{eqnarray}\label{Eq:Ginflowerbound}
      T(\mathbb{Q}_f|x,u_x)
    +\int_{\mathbb{Q}\setminus\mathbb{Q}_f}T(\mathbb{Q}_f|y,u_{y})T(dy|x,u_x)\geq \epsilon.
      \end{eqnarray}
\end{itemize}
\end{theorem}
\begin{proof}
  See Appendix C.
\end{proof}

\begin{remark}
  The value of $\rho$  is the largest probability that the next state $y$ remains outside the RCIS $\mathbb{Q}_f$ from any $x\in \mathbb{Q}\setminus\mathbb{Q}_f$ under the optimal stationary Markov policy in Lemma~\ref{TheAbata2}. Note that  $\frac{\rho^2}{1-\rho}$  is  the gap between the necessary condition and the sufficient condition. In addition, the second item in (\ref{Eq:Ginfupperbound})--(\ref{Eq:Ginflowerbound}) denotes the probability that the state is steered into the RCIS $\mathbb{Q}_f$ by two transitions from $x\in \mathbb{Q}\setminus\mathbb{Q}_f$ with an intermediate state $y$  outside $\mathbb{Q}_f$.
\end{remark}

\begin{corollary}\label{Cor:infinitenecessary}
Suppose that Assumption~\ref{infiniteassum} holds and let $0< \epsilon\leq1$ be fixed. A nonempty set $\mathbb{Q}$ is an infinite-horizon $\epsilon$-PCIS
\begin{itemize}
  \item[(i)] \emph{only if} there exists an RCIS $\mathbb{Q}_f\subseteq \mathbb{Q}$ such that $\forall x\in \mathbb{Q}\setminus\mathbb{Q}_f$, $T(\mathbb{Q}|x,u)\geq \epsilon$ for some $u\in \mathbb{U}$;
  \item[(ii)] \emph{if }there exists an RCIS $\mathbb{Q}_f\subseteq \mathbb{Q}$ such that $\forall x\in \mathbb{Q}\setminus\mathbb{Q}_f$, $T(\mathbb{Q}_f|x,u)+\epsilon T(\mathbb{Q}\setminus\mathbb{Q}_f|x,u)\geq \epsilon$ for some
$u\in \mathbb{U}$.
\end{itemize}
\end{corollary}
\begin{proof}
    See Appendix D.
\end{proof}

\begin{remark}\label{infinitecorosuffi}
  A nonempty set $\mathbb{Q}$ is an infinite-horizon $\epsilon$-PCIS if there exists an RCIS $\mathbb{Q}_f\subseteq \mathbb{Q}$  such that $\forall x\in \mathbb{Q}\setminus\mathbb{Q}_f$, $T(\mathbb{Q}_f|x,u)\geq \epsilon$ for some $u\in \mathbb{U}$. This implication will facilitate the design of an algorithm for an infinite-horizon $\epsilon$-PCIS, see Algorithm 4.
\end{remark}

\begin{remark}
Considering the similarity between the reliability defined in \cite{Hernandez(2017)} and the infinite-horizon invariance probability in this paper, we can extend the results on infinite-horizon PICSs, including the existence condition above and the computational algorithms in the following, to the reliable control set  in \cite{Hernandez(2016)} to general stochastic systems.
\end{remark}

\subsection{Infinite-horizon $\epsilon$-PCIS computation}
This subsection will address the following problem.
\begin{problem}\label{infiniteproblem}
Given a set $\mathbb{Q}\in \mathcal{B}(\mathbb{X})$ and a prescribed probability $0\leq \epsilon \leq 1$, compute an infinite-horizon $\epsilon$-PCIS $\tilde{\mathbb{Q}}\subseteq\mathbb{Q}$.
\end{problem}

To handle this problem, the key point is to compute the infinite-horizon invariance probability $G^*_{\infty,\mathbb{Q}}$. For discrete spaces, it is shown that computationally tractable MILP can be used to compute the exact value of $G^*_{\infty,\mathbb{Q}}$. In this case, we  can
compute the largest infinite-horizon $\epsilon$-PCIS by computing iteratively the stochastic backward reachable sets until convergence.  For continuous spaces, it is in general computationally intractable to compute $G^*_{\infty,\mathbb{Q}}$ and the discretization method fails to work since the approximation error in (\ref{difV1}) increases with the horizon. In this case, we design another computational algorithm based on the sufficient conditions in Remark~\ref{infinitecorosuffi}.

\subsubsection{Discrete state and control spaces}
 If the state and control
spaces are discrete, we adopt the same assumptions as in Section \ref{finitediscretespace}. We will first show how to compute the exact value of $G^*_{\infty,\mathbb{Q}}$ in (\ref{Ginf1})--(\ref{Ginf})  through an MILP.
Then, we will adapt Algorithm 1 to compute the largest infinite-horizon $\epsilon$-PCIS within a given set.

\paragraph{MILP reformulation}
  Since $0$ is a trivial solution of (\ref{Ginf}), we cannot directly reformulate (\ref{Ginf1})--(\ref{Ginf}) as an LP, which is the traditional way to deal with infinite-horizon stochastic optimal control problems~\cite{Bertsekas(2012)}.

The following lemma provides a computationally tractable MILP reformulation when computing $G^*_{\infty,\mathbb{Q}}$.
\begin{lemma}\label{lemmainf}
 Given any set $\mathbb{Q}\subseteq \mathbb{X}$, the value of $G^*_{\infty,\mathbb{Q}}$ in~(\ref{Ginf}) can be obtained by solving the MILP:
\begin{subequations}\label{MILPinfinite}
  \begin{eqnarray}
  && \hspace{-1.5cm}\max_{g(x),\kappa(x,u)} \ \sum_{x\in \mathbb{Q}}g(x) \\
  && \hspace{-1.5cm}\text{subject to} \ \forall x\in \mathbb{Q},\nonumber \\
  && \hspace{-1.3cm}
      g(x)\geq \sum_{y\in \mathbb{Q}}g(y)T(y|x,u),   \forall u\in \mathbb{U}_x,\label{great}\\
  && \hspace{-1.3cm}
      g(x)\leq \sum_{y\in \mathbb{Q}}g(y)T(y|x,u)+(1-\kappa(x,u))\Delta,   \forall u\in \mathbb{U}_x,\label{less}\\
  && \hspace{-1.3cm} \sum_{u\in\mathbb{U}_x}\kappa(x,u)\geq 1,\label{kappa}\\
  && \hspace{-1.3cm} 0 \leq g(x)\leq1, \kappa(x,u)\in \{0,1\},  \forall u\in \mathbb{U}_x,
   \end{eqnarray}
\end{subequations}
where $\Delta$ is a constant greater than one.
That is, $G^*_{\infty,\mathbb{Q}}(x)=g^*(x)$, $\forall x\in\mathbb{Q}$, where $g^*$ is the optimal solution of the MILP (\ref{MILPinfinite}). The optimal stationary Markov policy is $\bar{\mu}^*_{\mathbb{Q}}(x)=u$  where $u\in \mathbb{U}_x$ such that $\kappa^*(x,u)=1$ and $\kappa^*$ is the optimal solution of the MILP (\ref{MILPinfinite}).
\end{lemma}
\begin{proof}
 From the monotone decrease of the sequence $(G^*_{0,\mathbb{Q}},G^*_{1,\mathbb{Q}},\ldots)$ and Lemma~\ref{TheAbata2},  $G^*_{\infty,\mathbb{Q}}$ is the maximum fixed point satisfying (\ref{Ginf}).
  Hence, the equivalent form of $G^*_{\infty,\mathbb{Q}}$ can be written as MILP (\ref{MILPinfinite}), where the constraints (\ref{great})--(\ref{kappa}) guarantee that there exists $u\in \mathbb{U}_x$ such that the equality in (\ref{Ginf}) holds.
\end{proof}

\paragraph{Computational algorithm}
As an adaption of Algorithm 1, the following algorithm  provides a way to compute the largest infinite-horizon $\epsilon$-PCIS within $\mathbb{Q}$.
\begin{algorithm}[H]\label{algorithm3}
\caption{Infinite-horizon $\epsilon$-PCIS}
\begin{algorithmic}[1]
\State Initialize $i=0$ and $\mathbb{P}_i=\mathbb{Q}$.
\State Compute $G^*_{\infty,\mathbb{P}_i}(x)$ for all $x\in \mathbb{P}_i$.
\State Compute the set $\mathbb{P}_{i+1}=\mathbb{S}^{*}_{\epsilon,\infty}(\mathbb{P}_i)$.
\State If $\mathbb{P}_{i+1}=\mathbb{P}_i$, stop. Else, set $i=i+1$ and go to step 2.
\end{algorithmic}
\end{algorithm}

The difference between Algorithms~$1$ and $3$ is that the value of $G^*_{\infty,\mathbb{P}_i}(x)$, instead of $V^*_{0,\mathbb{P}_i}(x)$, $\forall x\in \mathbb{P}_i$,  is computed by (\ref{MILPinfinite}) (replacing $\mathbb{Q}$ with $\mathbb{P}_i$).  Furthermore, the updated set $\mathbb{P}_{i+1}=\mathbb{S}^{*}_{\epsilon,\infty}(\mathbb{P}_i)$, which is a stochastic backward reachable set within $\mathbb{P}_i$ with respect to infinite horizon and a probability level $\epsilon$. The following theorem provides the convergence of $\mathbb{P}_i$ and shows that the resulting set $\tilde{\mathbb{Q}}$ by this algorithm is an infinite-horizon $\epsilon$-PCIS.

\begin{theorem}\label{infinitemaximalPCIS}
    For discrete state and control spaces, Algorithm~3 converges in a finite number of iterations and generates the largest infinite-horizon $\epsilon$-PCIS within $\mathbb{Q}$.  Furthermore, at each iteration, the infinite-horizon invariance probability $G^*_{\infty,\mathbb{P}_i}(x)$, $\forall x\in \mathbb{P}_i$, can be computed via   the MILP (\ref{MILPinfinite}).
\end{theorem}
\begin{proof}
  The finite-step convergence of Algorithm~3 follows from the finite cardinality of the set $\mathbb{Q}$. Similar to Theorem~\ref{finitemaxPCIS1}, the generated infinite-horizon $\epsilon$-PCIS is the largest one within $\mathbb{Q}$. The MILP reformulation refers to Lemma~\ref{lemmainf}.
\end{proof}

\begin{remark}
  When implementing Algorithm 3 to a system with discrete spaces, the maximal iteration number is $|\mathbb{Q}|$. An MILP is used to compute the value of $G^*_{\infty,\mathbb{P}_i}(x)$, $\forall x\in \mathbb{P}_i$, at each iteration. The number of  real-valued decision values is at most $|\mathbb{Q}|$, the number of  binary decision values is at most $|\mathbb{Q}||\mathbb{U}|$, and the number of  constraints is at most $|\mathbb{Q}|(2|\mathbb{U}|+3)$. In general, MILPs are NP-hard and can be solved by cutting plane algorithm or branch-and-bound algorithm~\cite{Papadimitriou(1998)}. Some advanced softwares have been developed to solve large MILPs efficiently~\cite{Linderoth(2010),Bonami(2012)}.
\end{remark}

\subsubsection{Continuous state and control spaces}
If the state and control spaces are continuous, it is computationally intractable to compute the exact value of infinite-horizon invarinace probability $G^*_{\infty,\mathbb{Q}}(x)$.
Based on Remark~\ref{infinitecorosuffi}, this subsection provides another way to compute an infinite-horizon $\epsilon$-PCIS within a given set $\mathbb{Q}$.

Different from Algorithm 3, which computes iteratively the stochastic backward reachable sets, the following algorithm generates an infinite-horizon $\epsilon$-PCIS by computing a backward stochastic reachable set from the RCIS $\mathbb{Q}_f$ contained in $\mathbb{Q}$.
\begin{algorithm}[H]\label{algorithm4}
\caption{Infinite-horizon $\epsilon$-PCIS}
\begin{algorithmic}[1]
\State Compute the RCIS within $\mathbb{Q}$, denoted by $\mathbb{Q}_f$.
\State Compute the stochastic backward reachable set from $\mathbb{Q}_f$, i.e.,
$\tilde{\mathbb{Q}}=\{x\in \mathbb{Q} \mid \exists u\in \mathbb{U}, \int_{\mathbb{Q}_f}T(dy|x,u)\geq \epsilon \}$.
\end{algorithmic}
\end{algorithm}

The first step in Algorithm 4 is the computation of RCIS within a given set, which is a well-studied topic in the literature \cite{Rakovic(2005),Rungger(2017),Gilbert(1991)}. Then, based on RCIS $\mathbb{Q}_f$ within $\mathbb{Q}$, the stochastic backward reachable set
\begin{eqnarray*}
\tilde{\mathbb{Q}}=\{x\in \mathbb{Q} \mid \exists u\in \mathbb{U}, \int_{\mathbb{Q}_f}T(dy|x,u)\geq \epsilon \}
\end{eqnarray*}
 is an infinite-horizon $\epsilon$-PCIS within $\mathbb{Q}$. In comparision with Algorithms~1--3, the iteration is avoided in Algorithm~4, which only needs two steps.

\begin{remark}
Note that the resulting set from Algorithm~4 is in general not the largest infinite-horizon $\epsilon$-PCIS within the given set $\mathbb{Q}$. It is possible to obtain a larger infinite-horizon $\epsilon$-PCIS if we can  reformulate the existence conditions in Theorem~\ref{infinitenecessary} and Corollary~\ref{Cor:infinitenecessary} in a recursive form and thereby modify Algorithm 4 to be a recursive algorithm.
\end{remark}

\begin{remark}
The complexity of Algorithm 4 depends on the computation of the RCIS \cite{Blanchini(2007),Rakovic(2005),Rungger(2017),Gilbert(1991)}, and the computation of the backward stochastic reachable set. The later can be reformulated as a chance-constrained problem and then approximately solved.
      Some results on computation of the backward stochastic reachable set have been reported in \cite{Prandini(2006)}.
The first example in Section \ref{example} will show how to compute the backward stochastic reachable set.
\end{remark}

\begin{figure}[t]
\centering
	\subfigure{
	\includegraphics[width=0.32\textwidth]{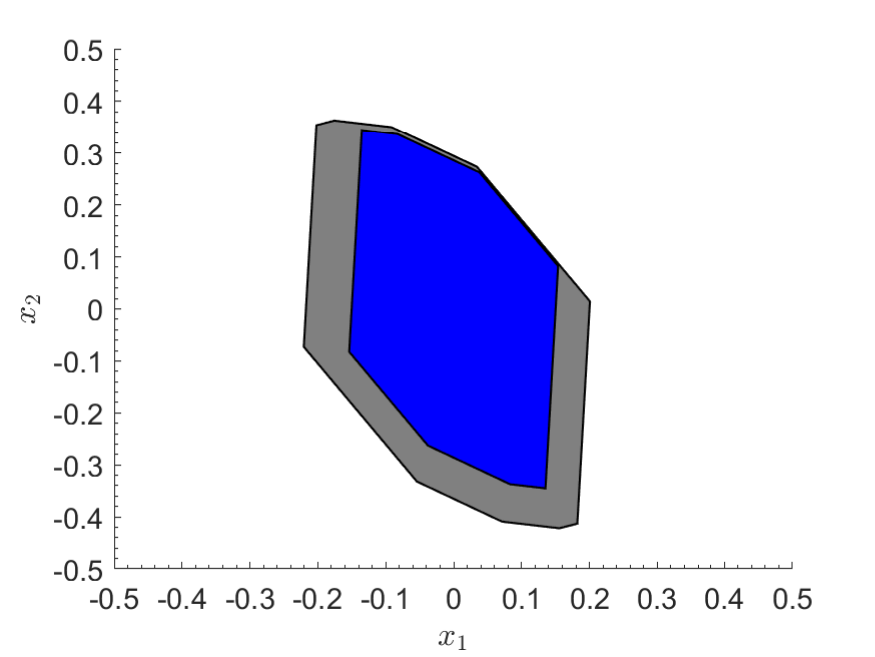}}
	\caption{\footnotesize Computations of the largest RCIS (blue) and an infinite-horizon $\epsilon$-PCIS with $\epsilon=0.80$ (gray) by Algorithm 4 for Example 1.}
	\label{mpc1}
\end{figure}

\begin{figure*}
\centering
\subfigure{
	\includegraphics[width=0.8\textwidth]{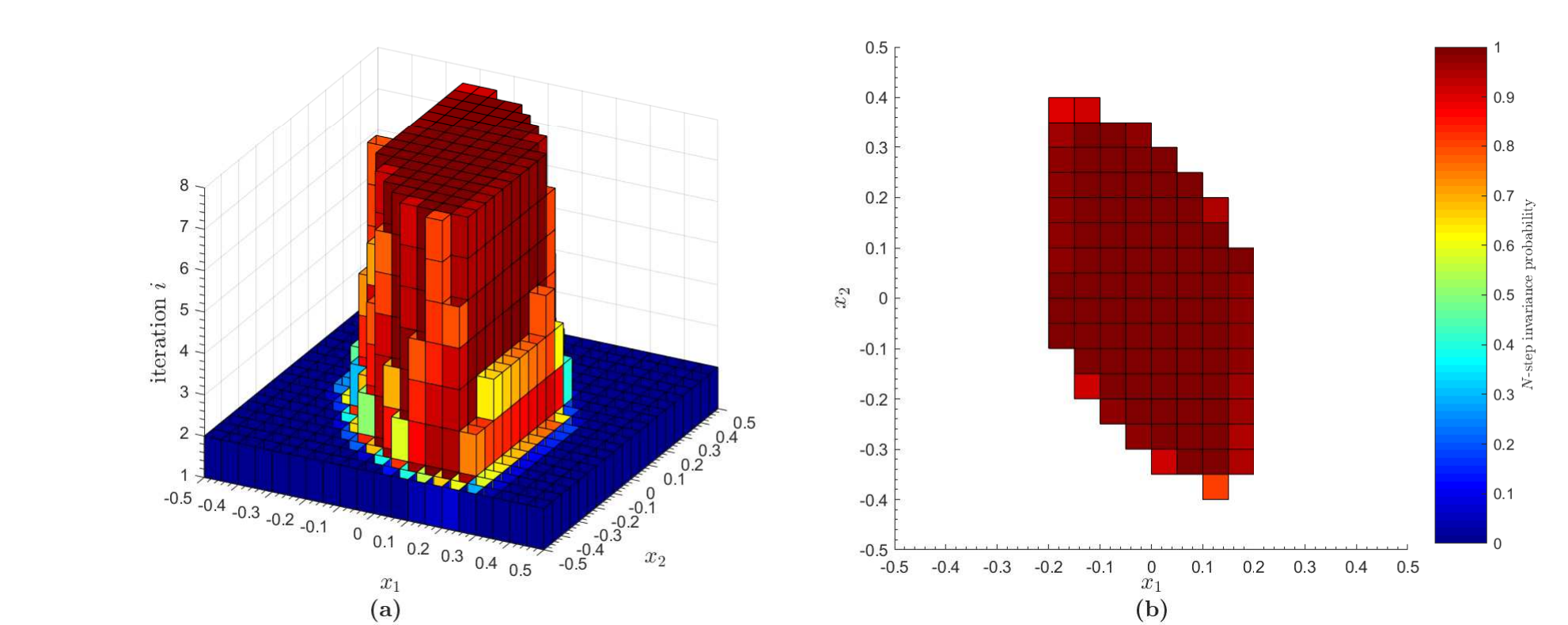}}
	\caption{\footnotesize  Computation of $N$-step $\epsilon$-PCIS with $N=5$ and $\epsilon=0.80$ for Example 1: (a) The sets $\mathbb{P}_i$ and the corresponding $N$-step invariance probability in Algorithm 2. (b) The $N$-step $\epsilon$-PCIS $\tilde{\mathbb{Q}}$.}
\label{mpc2}
\end{figure*}

\section{Examples}\label{example}
In this section, two examples are provided to illustrate the effectiveness of the proposed theoretical results. The first one is concerned with comparison between PCIS and RCIS. Then we consider an application to motion planning of a mobile robot in a partitioned space with obstacles.

\subsection{Example 1: Comparison between PCIS and RCIS}
Consider the following example from \cite{Kouvaritakis(2010)}:
\begin{equation*}
x_{k+1}=Ax_k+Bu_k+w_k,
\end{equation*}
where $A=\left [\begin{array}{ccccc}
1.6 & 1.1\\
-0.7 & 1.2
\end{array}\right]$ and $B=\left [\begin{array}{ccccc}
1\\
1
\end{array}\right]$.
The control input is constrained by $ |u_k|\leq 0.25$. We consider $w_k$ to be either  non-stochastic or stochastic when computing RCIS and PCIS, respectively. The region of interest
is~$\mathbb{Q}=\{x\in\mathbb{R}^2\mid \|x\|_{\infty}\leq 0.5\}$. We will compare the largest RCIS and PCIS within~$\mathbb{Q}$.

To derive an RCIS for this system, we assume the disturbance belongs to  the compact set $\mathbb{W}=\{w\in\mathbb{R}^2\mid \|w\|_{\infty}\leq 0.05\}$. By using the methods in \cite{Bertsekas(1972),Gilbert(1991)}, we obtain the  largest RCIS, which is the blue region shown in Fig.~\ref{mpc1}. The gray region is an infinite-horizon $\epsilon$-PCIS described in the end of this example.

When computing a finite-horzion PCIS, assume that elements of $w_k$ are  i.i.d. Gaussian random variables with zero mean and variance $\sigma^2=1/30^2$.
This system can be represented as a triple $\mathcal{S}=\{\mathbb{X},\mathbb{U},T\}$:
\begin{eqnarray*}
\begin{cases}
  \mathbb{X}=\mathbb{R}^2,\\
  \mathbb{U}=\{u\in \mathbb{R}\mid  |u|\leq 0.25\}, \\
  t(x_{k+1}|x_k,u_k)=\psi(\Lambda^{-1}(x_{k+1}-Ax_k-Bu_k)),
\end{cases}
\end{eqnarray*}
where $\psi(\cdot)$ is the density function of the standard normal distribution and $\Lambda=\rm{diag}\{\sigma,\sigma\}$. In this case, since the  Lipschitz constant $L$ in Assumption~\ref{assumt} is small, we ignore the approximation error $\tau_0$ in (\ref{tauk}).
We discretize the continuous spaces and implement Algorithm $2$ to compute the $N$-step $\epsilon$-PCIS $\tilde{\mathbb{Q}}$. First consider  $N=5$ and $\epsilon=0.80$. Fig.~\ref{mpc2}(a) shows the evolution of the set $\mathbb{P}_i$ in Algorithm 2. The color indicates the corresponding $N$-step invariance probability $p^{*}_{N,\mathbb{P}_i}(x)$ and the $z$-axes the iteration index $i$. The algorithm converges in $8$ steps. Fig.~\ref{mpc2}(b) shows $\mathbb{P}_8$, which corresponds to the $N$-step $\epsilon$-PCIS $\tilde{\mathbb{Q}}$ for  $N=5$ and $\epsilon=0.80$.

When computing an infinite-horizon PCIS, we choose the same bound on the disturbance as for the RCIS. The elements of $w_k$ are  truncated i.i.d. Gaussian random variables with zero mean and variance $\sigma^2=1/30^2$. Denote the largest RCIS computed above by $\mathbb{Q}_f=\{x\in \mathbb{R}^2\mid Hx\leq h\}$, where the matrix $H$ and the vector $h$ are with appropriate dimensions. As stated in Algorithm~$4$,  the one-step stochastic backward reachable set from the RCIS associated with probability $0.80$ is an infinite-horizon $\epsilon$-PCIS with $\epsilon=0.80$, i.e.,
  \begin{eqnarray*}
\tilde{\mathbb{Q}}=\{x\in \mathbb{Q} \mid \exists u\in \mathbb{U}, {\rm{Pr}}\{H(Ax+Bu+w)\leq h\}\geq 0.80 \}.
\end{eqnarray*}
This set can be represented as
  \begin{eqnarray*}
\tilde{\mathbb{Q}}=\{x\in \mathbb{Q} \mid \exists u\in \mathbb{U},H(Ax+Bu)+h'\leq h \},
\end{eqnarray*}
where $h'$ is the optimal solution of the chance constrained program
\begin{eqnarray*}
&&\hspace{0.2cm}\min \sum_{j}h'_j \\
&&\text{subject to} \ {\rm{Pr}}\{Hw\leq h'\}=0.8.
\end{eqnarray*}
This program can be numerically solved  by using the methods in \cite{Lorenzen(2017),Prekopa(2013)}. The resulting infinite-horizon $\epsilon$-PCIS with $\epsilon=0.80$ is the gray region shown in Fig.~\ref{mpc1}. This region is obviously a superset of the RCIS in blue.

%
%
%

\begin{figure}[t]
\centering
	\subfigure{
	\includegraphics[width=0.35\textwidth]{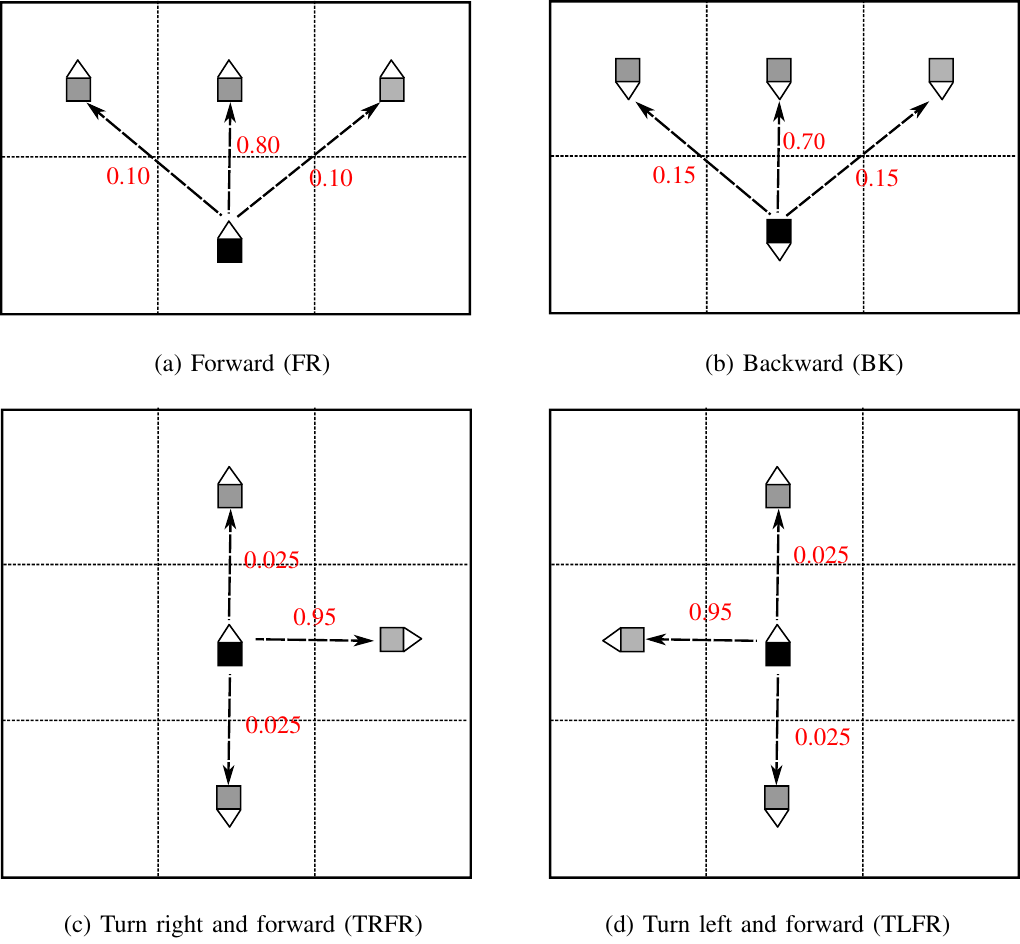}}
	\caption{\footnotesize Transition probability under actions for Example 2.}
\label{transitionpro}
\end{figure}

\begin{figure}[t]
\centering
	\subfigure{
	\includegraphics[width=0.35\textwidth]{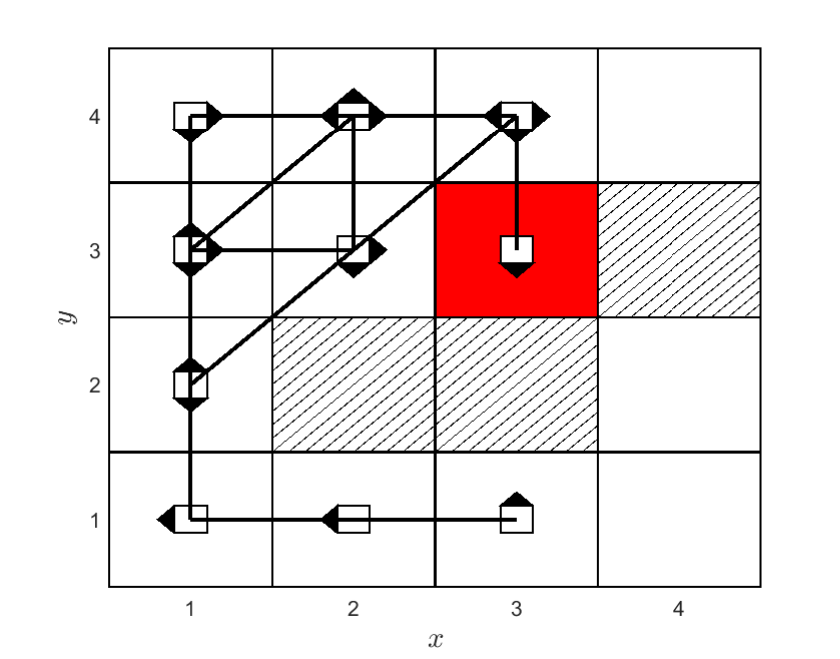}}
	\caption{\footnotesize One simulated state trajectory  with indication of the robot orientation starting from $(3,1,\mathcal{N})$ and ending at $(3,4,\mathcal{S})$ in Example 2.}
	\label{motion2}  
\end{figure}

\begin{figure*}
\hspace{-1.5cm}
	\subfigure{
	\includegraphics[width=1\textwidth]{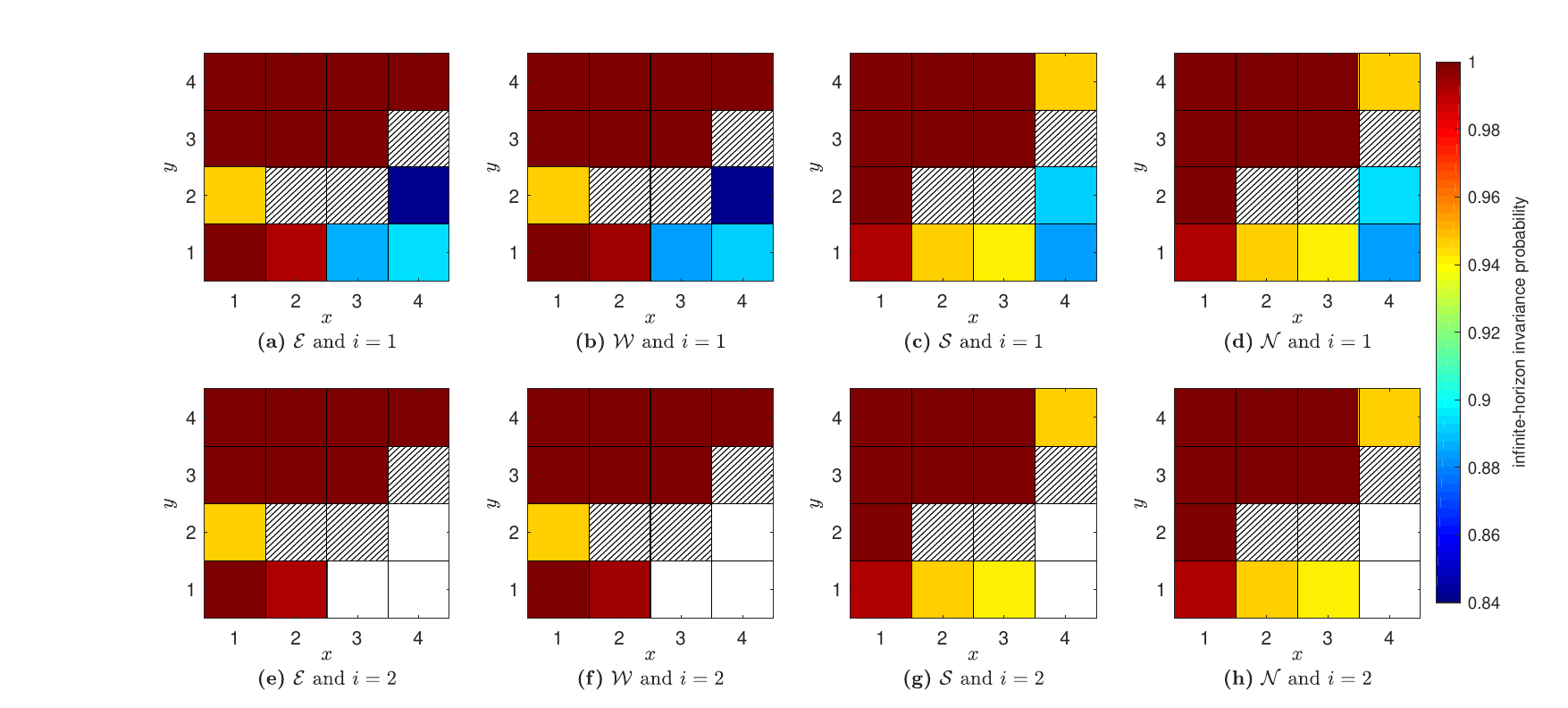}}
	\caption{\footnotesize The sets $\mathbb{P}_i$ and the corresponding infinite-horizon invariance probability in Example 2 when computing the largest infinite-horizon $\epsilon$-PCIS with $\epsilon=0.90$ by Algorithm 3.}
\label{motion}
\end{figure*}

\subsection{Example 2: Motion planning}
The motion planning example in \cite{GuoM(2018)} is adapted to seek an infinite-horizon PCIS within the workspace for a mobile robot.
The state of the robot is abstracted by its cell coordinate, i.e., $(p_x,p_y)\in\{1,2,3,4\}^2$, and its four possible orientations $\{\mathcal{E},\mathcal{W},\mathcal{S},\mathcal{N}\}$.  Due to the actuation noise and drifting, the robot motion is stochastic. Here, we restrict the action space to be $\{\rm{FR},\rm{BK},\rm{TRFR},\rm{TLFR}\}$, under which the possible transitions are shown in Fig.~\ref{transitionpro}. Specifically, action ``$\rm{FR}$" means driving forward for $1$ unit. As illustrated in the figure, the probability for that is $0.80$. The probability of drifting forward to the left or the right by $1$ unit is $0.10$. Action ``$\rm{BK}$" can be similarly defined. Action ``$\rm{TRFR}$" means turning right $\pi/2$ and driving forward for $1$ unit, of which the probability is $0.95$. The probability of driving forward for $1$ unit without turning right is $0.025$ and the probability of turning right for $\pi$ and  driving forward for $1$ unit is  $0.025$. Similarly, we can define the action ``$\rm{TLFR}$".

Consider the partitioned workspace shown in Fig.~\ref{motion2}, where the shadowed cells are occupied by obstacles and the red cell is an absorbing region, i.e., when the robot enters in this region it  will stay there forever.
We construct an MDP with $64$ states and $4$ actions. The transition relation and probability can be defined based on the above description. We compute the largest infinite-horizon $\epsilon$-PCIS with $\epsilon=0.90$ within the safe state space, i.e., the remaining of the state space by excluding the states associated with the obstacles.

By implementing Algorithm~$3$, the computed sets $\mathbb{P}_i$ and the corresponding infinite-horizon invariance probability $p^{*}_{\infty,\mathbb{P}_i}(x)$ are shown in Fig.~\ref{motion}, of which each subfigure corresponds to one orientation in $\{\mathcal{E},\mathcal{W},\mathcal{S},\mathcal{N}\}$. The first row of Fig.~\ref{motion} shows the  results after the first iteration, where we  can see that the infinite-horizon invariance probability $p^{*}_{\infty,\mathbb{P}_i}(x)$ at  $x=(4,2,\mathcal{E})$ and $x=(4,2,\mathcal{W})$ is less than $\epsilon=0.90$.
Algorithm~$3$ converges in $2$ steps and generates the largest infinite-horizon $\epsilon$-PCIS $\tilde{\mathbb{Q}}$ with $\epsilon=0.90$ shown in Fig.~\ref{motion}(e)--\ref{motion}(h). This invariant set provides a region where the admissible action can drive the robot without colliding with the obstacles with  probability $0.90$. By implementing the optimal policy obtained in Lemma~\ref{lemmainf}, we run a state trajectory starting from $(3,1,\mathcal{N})$ as shown in Fig.~\ref{motion2}. We can see that this trajectory is collision-free and finally ends at the absorbing region $(3,3,\mathcal{S})$.

\section{Conclusion}\label{conclusion}
We investigated the extension of set invariance in a stochastic sense for control systems. We proposed  finite- and infinite-horizon $\epsilon$-PCISs, and provided some fundamental properties. We designed
iterative algorithms to compute the PCIS within a given set. For systems
with discrete state and control spaces,  finite- and infinite-horizon $\epsilon$-PCISs can be computed by solving an
LP and an MILP at each iteration, respectively. We proved that the iterative algorithms were
computationally tractable and can be terminated in a finite number
of steps. For systems with continuous state and
control spaces, we established the approximation of
stochastic control systems and proved its convergence when computing finite-horizon $\epsilon$-PCIS.  In addition, thanks to the sufficient conditions for the existence of infinite-horizon $\epsilon$-PCIS, we can compute an infinite-horizon $\epsilon$-PCIS  by the stochastic backward reachable set from the RCIS contained in it. Numerical
examples were given to illustrate the theoretical results.

One future direction is to apply the PCISs to safety-critical control and stochastic predictive control. In particular, how to characterize  stability  using PCISs is an important problem to consider. Another interesting  future extension of PCISs is  to study  reliability and mean-time-to-failure for general stochastic systems.

\section*{Acknowledgment}
The authors are grateful to Prof. Alessandro Abate for helpful discussions and feedback and to anonymous reviewers for their constructive comments.

\section*{Appendix A. Proof of Lemma \ref{finiteuniver}}
Define the functions $J^*_{k,\mathbb{Q}}: \mathbb{X}\rightarrow \mathbb{R}$, $k\in\mathbb{N}_{[0,N]}$, as
\begin{eqnarray*}
J^*_{k,\mathbb{Q}}(x)=-V^*_{N-k,\mathbb{Q}}(x), \forall x\in \mathbb{X}.
\end{eqnarray*}
As shown in \cite{Abate(2007)}, the function $J^*_{N,\mathbb{Q}}$ is lower-semianalytic for any $\mathbb{Q}\in \mathcal{B}(\mathbb{X})$. From Definitions 7.20 and 7.21 in \cite{Bertsekas(2004)}, we have that the function $J^*_{N,\mathbb{Q}}$ is also analytically measurable and thus is universally measurable for any $\mathbb{Q}\in \mathcal{B}(\mathbb{X})$. According to the definition of universal measurability, the set $ J^{*,-1}_{N,\mathbb{Q}}(\mathbb{B})=\{x\in \mathbb{X}\mid J^*_{k,\mathbb{Q}}(x)\in \mathbb{B}\}$ for $\mathbb{B}\in \mathcal{B}(\mathbb{R})$ is universally measurable.

Recall the definition of the stochastic backward reachable set $\mathbb{S}^{*}_{\epsilon,N}(\mathbb{Q})$, we have that \begin{eqnarray*}
 &&\hspace{0cm}\mathbb{S}^{*}_{\epsilon,N}(\mathbb{Q})=\{x\in\mathbb{Q}\mid V^*_{0,\mathbb{Q}}(x)\geq \epsilon\}\\
&&\hspace{1.22cm} =\{x\in\mathbb{Q}\mid -1 \leq J^*_{N,\mathbb{Q}}(x)\leq -\epsilon\}\\
&&\hspace{1.22cm} =J^{*,-1}_{N,\mathbb{Q}}(\mathbb{B})
 \end{eqnarray*}
where $\mathbb{B}=[-1,-\epsilon]\in \mathcal{B}(\mathbb{R})$.
Thus, the set $\mathbb{S}^{*}_{\epsilon,N}(\mathbb{Q})$ is universally measurable for any $\mathbb{Q}\in \mathcal{B}(\mathbb{X})$.
\section*{Appendix B. Proof of Lemma \ref{Thedisfinite}}\label{AppB}
Before proving Lemma \ref{Thedisfinite}, we need two auxiliary lemmas.
Lemma \ref{LipschitzV} shows that the value functions in (\ref{Vk}) are Lipschitz continuous. It is adapted from Theorem~8 in  \cite{Abate(2007)}.   Lemma \ref{Lipschitzt} shows that the difference between the approximate density function and the original density function is bounded.

\begin{lemma}\label{LipschitzV}
  Under Assumptions~\ref{finiteassum} and \ref{assumt}, for any $x,x'\in \mathbb{Q}$, the value functions $V^*_{k,\mathbb{Q}}$ in (\ref{Vk}) satisfy
  \begin{eqnarray}\label{Wineq}
  |V^*_{k,\mathbb{Q}}(x)-V^*_{k,\mathbb{Q}}(x')|\leq \phi(\mathbb{Q})L\|x-x'\|, \forall k\in \mathbb{N}_{[0,N]}.
  \end{eqnarray}
\end{lemma}
\begin{proof}
Similar to Theorem~8 in  \cite{Abate(2007)}.
\end{proof}

\begin{lemma}\label{Lipschitzt}
  Under Assumptions \ref{assumt}, for all $y\in \mathbb{Q}$ and $q_i\in\hat{\mathbb{Q}}$,
  \begin{eqnarray*}\label{dift}
\int_{\mathbb{Q}}|\hat{t}(y|q_i,\hat{u})-t(y|q_i,\hat{u})|dy\leq 2\phi(\mathbb{Q})L\delta, \forall \hat{u}\in \hat{\mathbb{U}}.
\end{eqnarray*}
\end{lemma}
\begin{proof}
  If $\int_{\mathbb{Q}}t(s_z|s_x,\hat{u})dz< 1$, it follows from Assumption \ref{assumt} that
  \begin{eqnarray*}
  \int_{\mathbb{Q}}|\hat{t}(y|q_i,\hat{u})-t(y|q_i,\hat{u})|dy\leq \phi(\mathbb{Q})L\delta.
  \end{eqnarray*}
  And if $\int_{\mathbb{Q}}t(s_z|s_x,\hat{u})dz\geq 1$, we first have
  \begin{eqnarray*}
  &&0\leq\int_{\mathbb{Q}}t(s_y|q_i,\hat{u})dy-1\\
  &&\hspace{0.25cm}\leq \int_{\mathbb{Q}}t(s_y|q_i,\hat{u})dy-
  \int_{\mathbb{Q}}t(y|q_i,\hat{u})dy \\
  &&\hspace{0.25cm}\leq\int_{\mathbb{Q}}|t(s_y|q_i,\hat{u})
  -t(y|q_i,\hat{u})|dy\\
  &&\hspace{0.25cm}\leq\phi(\mathbb{Q})L\delta.
  \end{eqnarray*}
  Furthermore, we have
  \begin{eqnarray*}
  &&\int_{\mathbb{Q}}|\hat{t}(y|q_i,\hat{u})-t(y|q_i,\hat{u})|dy\\
  &&\hspace{-0.4cm}=\int_{\mathbb{Q}}\frac{|t(s_y|q_i,\hat{u})
  -t(y|q_i,\hat{u})\int_{\mathbb{Q}}t(s_z|s_x,\hat{u})dz|}
  {\int_{\mathbb{Q}}t(s_z|s_x,\hat{u})dz}dy\\
  &&\hspace{-0.4cm}\leq \int_{\mathbb{Q}}|t(s_y|q_i,\hat{u})
  -t(y|q_i,\hat{u})\int_{\mathbb{Q}}t(s_z|s_x,\hat{u})dz|dy\\
  &&\hspace{-0.4cm}\leq \int_{\mathbb{Q}}|t(s_y|q_i,\hat{u})-t(y|q_i,\hat{u})|dy+ \\
  &&|\int_{\mathbb{Q}}t(s_z|s_x,\hat{u})dz-1|\int_{\mathbb{Q}}|t(y|q_i,\hat{u})|dy\\
  &&\hspace{-0.4cm}\leq 2\phi(\mathbb{Q})L\delta.
  \end{eqnarray*}
  This completes the proof.
\end{proof}

\emph{Proof of Lemma \ref{Thedisfinite}}:
  First of all, let us prove the inequality (\ref{difV1}). It is easy to check it for $k=N$ since $V^*_{N,\mathbb{Q}}(x)=\hat{V}^*_{k,\mathbb{Q}}(x)=1, \forall x\in \mathbb{Q}$. By induction, we assume that $|V^*_{k+1,\mathbb{Q}}(x)-\hat{V}^*_{k+1,\mathbb{Q}}(x)|\leq \tau_{k+1}(\mathbb{Q})\delta$, $x\in \mathbb{Q}$. For any $q_i\in \mathbb{Q}_i$, $i\in \mathbb{N}_{[1,m_x]}$, we define $\mu^*_k=\arg\sup_{u\in\mathbb{U}}\int_{\mathbb{Q}}V^*_{k+1,\mathbb{Q}}(y)t(y|q_i,u)dy$
   and $\hat{\mu}^*_k=\arg \max_{\hat{u}\in \hat{\mathbb{U}}}\int_{\mathbb{Q}}\hat{V}^*_{k+1,\mathbb{Q}}(y)\hat{t}(y|q_i,\hat{u})dy$.
According to the dicretization procedure of the control space, we can choose some  $\hat{\nu}_k\in \hat{\mathbb{U}}$ such that $\|\mu^*_k-\hat{\nu}_k\|\leq \delta$.
Then, we have that
  \begin{eqnarray*}
  &&\hspace{0cm}V^*_{k,\mathbb{Q}}(q_i)-\hat{V}^*_{k,\mathbb{Q}}(q_i) \\
  &&\hspace{-0.4cm} = \int_{\mathbb{Q}}V^*_{k+1,\mathbb{Q}}(y)t(y|q_i,\mu^*_k)dy -\int_{\mathbb{Q}}
\hat{V}^*_{k+1,\mathbb{Q}}(y)\hat{t}(y|q_i,\hat{\mu}^*_k)dy \\
 &&\hspace{-0.4cm} \leq \int_{\mathbb{Q}}V^*_{k+1,\mathbb{Q}}(y)t(y|q_i,\mu^*_k)dy -\int_{\mathbb{Q}}
\hat{V}^*_{k+1,\mathbb{Q}}(y)\hat{t}(y|q_i,\hat{\nu}_k)dy\\
&&\hspace{-0.4cm} \leq |\int_{\mathbb{Q}}V^*_{k+1,\mathbb{Q}}(y)t(y|q_i,\mu^*_k)dy -\int_{\mathbb{Q}}V^*_{k+1,\mathbb{Q}}(y)t(y|q_i,\hat{\nu}_k)dy|+ \\
&&\hspace{0cm}|\int_{\mathbb{Q}}V^*_{k+1,\mathbb{Q}}(y)t(y|q_i,\hat{\nu}_k)dy -\int_{\mathbb{Q}}V^*_{k+1,\mathbb{Q}}(y)\hat{t}(y|q_i,\hat{\nu}_k)dy|+\\
&&\hspace{0cm}|\int_{\mathbb{Q}}V^*_{k+1,\mathbb{Q}}(y)\hat{t}(y|q_i,\hat{\nu}_k)dy -\int_{\mathbb{Q}}\hat{V}^*_{k+1,\mathbb{Q}}(y)\hat{t}(y|q_i,\hat{\nu}_k)dy|\\
&&\hspace{-0.4cm}  \leq\phi(\mathbb{Q})L\delta + 2\phi(\mathbb{Q})L\delta+\tau_{k+1}(\mathbb{Q})\delta\\
&&\hspace{-0.4cm}=(3\phi(\mathbb{Q})L+\tau_{k+1}(\mathbb{Q}))\delta,
  \end{eqnarray*}
and
  \begin{eqnarray*}
  &&\hspace{0cm}\hat{V}^*_{k,\mathbb{Q}}(q_i)-V^*_{k,\mathbb{Q}}(q_i)\\
  &&\hspace{-0.4cm} \leq \int_{\mathbb{Q}}\hat{V}^*_{k+1,\mathbb{Q}}(y)\hat{t}(y|q_i,\hat{\mu}^*_k)dy -\int_{\mathbb{Q}}V^*_{k+1,\mathbb{Q}}(y)t(y|q_i,\hat{\mu}^*_k)dy\\
  &&\hspace{-0.4cm} \leq |\int_{\mathbb{Q}}\hat{V}^*_{k+1,\mathbb{Q}}(y)\hat{t}(y|q_i,\hat{\mu}^*_k)dy
  -\int_{\mathbb{Q}}\hat{V}^*_{k+1,\mathbb{Q}}(y)t(y|q_i,\hat{\mu}^*_k)dy|+ \\
  && |\int_{\mathbb{Q}}\hat{V}^*_{k+1,\mathbb{Q}}(y)t(y|q_i,\hat{\mu}^*_k)dy
  -\int_{\mathbb{Q}}V^*_{k+1,\mathbb{Q}}(y)t(y|q_i,\hat{\mu}^*_k)dy|\\
  &&\hspace{-0.4cm}\leq (2\phi(\mathbb{Q})L+\tau_{k+1}(\mathbb{Q}))\delta.
  \end{eqnarray*}
  Thus, we have
  \begin{eqnarray*}
  |V^*_{k,\mathbb{Q}}(q_i)-\hat{V}^*_{k,\mathbb{Q}}(q_i)|\leq (3\phi(\mathbb{Q})L+\tau_{k+1}(\mathbb{Q}))\delta.
  \end{eqnarray*}
  For any $x\in \mathbb{Q}_i$, $i\in \mathbb{N}_{[1,m_x]}$, it follows that
  \begin{eqnarray*}
   &&|V^*_{k,\mathbb{Q}}(x)-\hat{V}^*_{k,\mathbb{Q}}(x)|\\
   &&\hspace{-0.4cm}=|V^*_{k,\mathbb{Q}}(x)-\hat{V}^*_{k,\mathbb{Q}}(q_i)|\\
   &&\hspace{-0.4cm}\leq|V^*_{k,\mathbb{Q}}(x)-V^*_{k,\mathbb{Q}}(q_i)|+|V^*_{k,\mathbb{Q}}(q_i)-\hat{V}^*_{k,\mathbb{Q}}(q_i)|\\
   &&\hspace{-0.4cm}\leq  (4\phi(\mathbb{Q})L+\tau_{k+1}(\mathbb{Q}))\delta= \tau_{k}(\mathbb{Q})\delta,
  \end{eqnarray*}
  which completes the proof of the inequality (\ref{difV1}).

\section*{Appendix C. Proof of Theorem \ref{infinitenecessary}}\label{AppC}
 Let $u_x$ be the control input such that (\ref{Ginf}) holds for any $x\in \mathbb{Q}$.

\emph{Only-if-part}: Under Assumption~\ref{infiniteassum}, the fact that the set $\mathbb{Q}\in\mathcal{B}(\mathbb{X})$ is an infinite-horizon $\epsilon$-PCIS is equivalent to $G^*_{\infty,\mathbb{Q}}(x)\geq \epsilon, \forall x\in \mathbb{Q}$. Let $\theta=\sup_{x\in\mathbb{Q}}G^*_{\infty,\mathbb{Q}}(x)$. Under Assumption~\ref{infiniteassum},  $G^*_{\infty,\mathbb{Q}}(x)$ exists for all $x\in \mathbb{Q}$. The set $\tilde{\mathbb{Q}}_f=\{x\in \mathbb{Q}\mid G^*_{\infty,\mathbb{Q}}(x)=\theta\}$ collects all the states for which the value of $G^*_{\infty,\mathbb{Q}}$ is maximal over the set $\mathbb{Q}$. Extending Lemma \ref{finiteuniver} to infinite-horizon case, we have that the set $\tilde{\mathbb{Q}}_f$ is universally measurable. By Lemma~7.16 in \cite{Bertsekas(2004)}, we have that for any $p\in \mathcal{P}(\mathbb{X})$, there exists a Borel-measurable set $\mathbb{Q}_f\subseteq \mathbb{Q}$ such that $p(\mathbb{Q}_f\bigtriangleup \tilde{\mathbb{Q}}_f)=0$.

 Next we will show that the set $\mathbb{Q}_f$ is an RCIS. It follows from Assumption~\ref{infiniteassum} and  Lemma~\ref{TheAbata2} that $\forall x\in \mathbb{Q}_f$,
\begin{eqnarray}
  &&\hspace{-0.8cm}G^*_{\infty,\mathbb{Q}}(x)\nonumber\\
  &&\hspace{-1.2cm}=\int_{\mathbb{Q}_f}G^*_{\infty,\mathbb{Q}}(y)T(dy|x,u_x)
  +\int_{\mathbb{Q}\setminus\mathbb{Q}_f}G^*_{\infty,\mathbb{Q}}(y)T(dy|x,u_x) \nonumber\\
  &&\hspace{-1.2cm}=G^*_{\infty,\mathbb{Q}}(x)\int_{\mathbb{Q}_f}T(dy|x,u_x)
  +\nonumber \\
  && \hspace{3cm}\int_{\mathbb{Q}\setminus\mathbb{Q}_f}G^*_{\infty,\mathbb{Q}}(y)T(dy|x,u_x) \label{Qfdef1}\\
  &&\hspace{-1.2cm}\leq G^*_{\infty,\mathbb{Q}}(x)T(\mathbb{Q}_f|x,u_x)
  +G^*_{\infty,\mathbb{Q}}(x)T(\mathbb{Q}\setminus\mathbb{Q}_f|x,u_x)\label{Qfdef2}\\
  &&\hspace{-1.2cm}=  G^*_{\infty,\mathbb{Q}}(x)(T(\mathbb{Q}_f|x,u_x)+T(\mathbb{Q}\setminus\mathbb{Q}_f|x,u_x)),\nonumber
  \end{eqnarray}
where Eq. (\ref{Qfdef1}) follows from $G^*_{\infty,\mathbb{Q}}(x)=G^*_{\infty,\mathbb{Q}}(y), \forall x, y\in \mathbb{Q}_f$ and  Eq. (\ref{Qfdef2}) follows from that $G^*_{\infty,\mathbb{Q}}(x)>G^*_{\infty,\mathbb{Q}}(y), \forall x\in {Q}_f, \forall y\in \mathbb{Q}\setminus\mathbb{Q}_f$. Furthermore, since $G^*_{\infty,\mathbb{Q}}(x)\geq \epsilon > 0, \forall x\in \mathbb{Q}$, and  $0\leq T(\mathbb{Q}|x,u_x)\leq 1$, the equality in Eq. (\ref{Qfdef2}) holds if and only if $T(\mathbb{Q}_f|x,u_x)= 1$ and thereby $T(\mathbb{Q}\setminus\mathbb{Q}_f|x,u_x))=0$. Based on the recursion in (\ref{Ginf1}), we have $G^*_{\infty,\mathbb{Q}}(x)=1, \forall x\in \mathbb{Q}_f$. Hence, the set $\mathbb{Q}_f\subseteq \mathbb{Q}$ is an RCIS.

Next let us prove that  $\forall x\in \mathbb{Q}\setminus\mathbb{Q}_f$, Eq.(\ref{Eq:Ginfupperbound}) holds. That is to prove that
 \begin{eqnarray}
 &&\hspace{-0.8cm} G^*_{\infty,\mathbb{Q}}(x)\leq T(\mathbb{Q}_f|x,u_x)
    +\int_{\mathbb{Q}\setminus\mathbb{Q}_f}T(\mathbb{Q}_f|y,u_y)T(dy|x,u_x)\nonumber\\
  &&\hspace{4cm}  +\frac{\rho^2}{1-\rho}.\label{Eq:Ginfupperbound1}
 \end{eqnarray}
 By Theorem~7 in \cite{Abate(2007)}, the control input $u_x$ is also optimal to the recursion (\ref{Ginf1}). For all $k\in \mathbb{N}$,  we have $\forall x\in \mathbb{Q}_f$,  $G^*_{k,\mathbb{Q}}(x)=1$ and $\forall x\in \mathbb{Q}\setminus\mathbb{Q}_f$,
\begin{eqnarray*}
G^*_{k+1,\mathbb{Q}}(x)=T(\mathbb{Q}_f|x,u_x)
    +\int_{\mathbb{Q}\setminus\mathbb{Q}_f}G^*_{k,\mathbb{Q}}(y)T(dy|x,u_x).
\end{eqnarray*}
Let $\rho=\sup_{x\in \mathbb{Q}\setminus\mathbb{Q}_f}\int_{\mathbb{Q}\setminus\mathbb{Q}_f}T(dy|x,u_x)$. Note that  $0\leq \rho<1$. Then, $\forall x\in \mathbb{Q}\setminus\mathbb{Q}_f$, we can follow the induction rule to prove that
 \begin{eqnarray*}
 &&\hspace{-0.8cm}  G^*_{k,\mathbb{Q}}(x)\leq T(\mathbb{Q}_f|x,u_x)
    +\int_{\mathbb{Q}\setminus\mathbb{Q}_f}T(\mathbb{Q}_f|y,u_y)T(dy|x,u_x)\nonumber\\
  &&\hspace{4cm}  +\frac{\rho^2-\rho^k}{1-\rho},
 \end{eqnarray*}
which by taking limitation yields that (\ref{Eq:Ginfupperbound1}) holds.

\emph{If-part}: The proof for the existence of an RCIS $\mathbb{Q}_f\subseteq \mathbb{Q}$ is the same as that of the only if part. As shown above, the condition $T(\mathbb{Q}_f|x,u_x)=1$ is equivalent to $G^*_{\infty,\mathbb{Q}}(x)=1, \forall x\in \mathbb{Q}_f$. We can use induction to prove that  $\forall x\in \mathbb{Q}\setminus\mathbb{Q}_f$,
 \begin{eqnarray*}
 &&\hspace{-0.8cm}  G^*_{k,\mathbb{Q}}(x)\geq T(\mathbb{Q}_f|x,u_x)
    +\int_{\mathbb{Q}\setminus\mathbb{Q}_f}T(\mathbb{Q}_f|y,u_y)T(dy|x,u_x),
 \end{eqnarray*}
which further implies that $G^*_{\infty,\mathbb{Q}}(x)\geq T(\mathbb{Q}_f|x,u_x)
    +\int_{\mathbb{Q}\setminus\mathbb{Q}_f}T(\mathbb{Q}_f|y,u_y)T(dy|x,u_x)$.
 One sufficient condition to guarantee  $G^*_{\infty,\mathbb{Q}}(x)\geq \epsilon$ is (\ref{Eq:Ginflowerbound}), i.e., $T(\mathbb{Q}_f|x,u_x)
    +\int_{\mathbb{Q}\setminus\mathbb{Q}_f}T(\mathbb{Q}_f|y,u_y)T(dy|x,u_x)\geq \epsilon$.
  The proof is completed.

\section*{Appendix D. Proof of Corollary~\ref{Cor:infinitenecessary}}
By Lemma~\ref{TheAbata2} and  Theorem \ref{infinitenecessary}, the necessary condition in Corollary~\ref{Cor:infinitenecessary} can be proven by showing that  $\forall x\in \mathbb{Q}\setminus\mathbb{Q}_f$, there exists a $u\in \mathbb{U}$ such that
\begin{eqnarray}
  &&\hspace{-1.5cm}\epsilon\leq G^*_{\infty,\mathbb{Q}}(x)=\int_{\mathbb{Q}_f}G^*_{\infty,\mathbb{Q}}(y)T(dy|x,u)+\nonumber\\
  &&\hspace{0.8cm}\int_{\mathbb{Q}\setminus\mathbb{Q}_f}G^*_{\infty,\mathbb{Q}}(y)T(dy|x,u) \nonumber\\
  &&\hspace{0.5cm}\leq T(\mathbb{Q}_f|x,u)
  +T(\mathbb{Q}\setminus\mathbb{Q}_f|x,u)\label{conditions2}\\
  &&\hspace{0.5cm}=  T(\mathbb{Q}|x,u), \nonumber
  \end{eqnarray}
  where Eq. (\ref{conditions2}) follows from $0< G^*_{\infty,\mathbb{Q}}(x)\leq 1, \forall x\in \mathbb{Q}$.

The sufficient condition in Corollary~\ref{Cor:infinitenecessary} can be proven by showing that $\forall x\in \mathbb{Q}\setminus\mathbb{Q}_f$, there exists a $u\in \mathbb{U}$
 \begin{eqnarray}
  &&\hspace{-0.3cm}G^*_{\infty,\mathbb{Q}}(x)\nonumber\\
  &&\hspace{-0.7cm}=\int_{\mathbb{Q}_f}G^*_{\infty,\mathbb{Q}}(y)T(dy|x,u)
  +\int_{\mathbb{Q}\setminus\mathbb{Q}_f}G^*_{\infty,\mathbb{Q}}(y)T(dy|x,u) \nonumber\\
  &&\hspace{-0.7cm}\geq T(\mathbb{Q}_f|x,u)
  +\epsilon T(\mathbb{Q}\setminus\mathbb{Q}_f|x,u),\label{conditions22}
  \end{eqnarray}
  where Eq. (\ref{conditions22}) follows from $G^*_{\infty,\mathbb{Q}}(x)\geq \epsilon >0, \forall x\in \mathbb{Q}$. One sufficient condition to guarantee  $G^*_{\infty,\mathbb{Q}}(x)\geq \epsilon$ is $T(\mathbb{Q}_f|x,u)
  +\epsilon T(\mathbb{Q}\setminus\mathbb{Q}_f|x,u)\geq \epsilon$.
  The proof is completed.

\end{document}